\documentclass{wscpaperproc}
\usepackage{latexsym}
\usepackage{changes}
\usepackage{comment}
\newif\ifshowmoved
\showmovedtrue
\ifshowmoved
  \specialcomment{movedout}{\par\begingroup\color{gray}}{\endgroup\par}
\else
  \excludecomment{movedout}
\fi
\newif\ifshowdeleted
\showdeletedtrue

\ifshowdeleted
  \specialcomment{deletedblock}{\par\begingroup\color{gray}}{\endgroup\par}
\else
  \excludecomment{deletedblock}
\fi
\usepackage{graphicx}
\usepackage{mathptmx}
\usepackage[T1]{fontenc}

\usepackage{amsmath}
\usepackage{amsfonts}
\usepackage{amssymb}
\usepackage{amsbsy}
\usepackage{amsthm}
\usepackage{mathtools}
\usepackage{bm}

\usepackage{multirow}
\usepackage{algorithmic}
\usepackage{algorithm}
\usepackage{booktabs}
\usepackage[inline]{enumitem}
\newtheoremstyle{wsc}% hnamei
{3pt}% hSpace abovei
{3pt}% hSpace belowi
{}% hBody fonti
{}% hIndent amounti1
{\bf}% hTheorem head fontbf
{}% hPunctuation after theorem headi
{.5em}% hSpace after theorem headi2
{}% hTheorem head spec (can be left empty, meaning `normal')i

\theoremstyle{wsc}
\newtheorem{theorem}{Theorem}

\newtheorem{proposition}[theorem]{Proposition}

\newtheorem{remark}{Remark}
\newtheorem{problem}{Problem}

\newcommand{\erdosrenyi}{Erd\H{o}s-R\'{e}nyi}
\newcommand{\er}{\text{ER}}
\newcommand{\sbm}{\text{SBM}}
\newcommand{\icmodel}{\text{IC}}
\newcommand{\reals}{\mathbb{R}}

\newcommand{\obs}{\mathcal{O}}
\newcommand{\obssamp}{\obs_\text{samp}}
\newcommand{\pr}{\text{Pr}}
\newcommand{\scn}{\mathcal{S}}
\newcommand{\inter}{\mathcal{I}}
\newcommand{\cset}{\mathcal{C}}
\newcommand{\gset}{\mathcal{G}}
\newcommand{\ic}{\text{IC}}
\newcommand{\diffusion}{\mathcal{F}}

\newcommand{\scnset}{{S}}
\newcommand{\cas}{C}
\newcommand{\module}{\mathcal{M}}
\newcommand{\prob}{\textsc{ScenarioID}}

\newcommand{\gnnmodel}{\mathcal{P}}
\newcommand{\gnnnumlayers}{\mathcal{T}}
\newcommand{\gnnmessage}{\phi}
\newcommand{\gnnaggregate}{\bigoplus}
\newcommand{\gnnupdate}{\psi}
\newcommand{\gnnreadout}{\beta}
\newcommand{\gnnsample}{\gamma}
\newcommand{\hiddenembs}{\mathbf{H}}

\newcommand{\nodeembs}{\mathbf{N}}
\newcommand{\classhead}{\mathcal{H}}
\newcommand{\graphembs}{{M}}
\newcommand{\dataset}{\mathcal{D}}

\newcommand{\noedge}{no-edge}
\newcommand{\epicurveandstructure}{\textsc{FeatEng}}

\newcommand{\popdata}{TN/VA}
\newcommand{\sbmstudy}{\textit{SBM-2comm}}
\newcommand{\sbmthreecomm}{\textit{SBM-3comm}}

\newcommand{\pt}{p_t}
\newcommand{\nc}{{n_C}}
\newcommand{\mc}{{m_C}}
\newcommand{\dc}{{\delta_C}}

\newcommand{\trainfunction}{{\tt TrainModel}}
\newcommand{\highlevelfunction}{{\tt HigherLevelTraining}}
\newcommand{\classheadout}{{Q}}

\newcommand{\communitymatch}{community-match}
\newcommand{\communityid}{community-id}

\newif\ifextended
\extendedtrue

\newcommand{\apptext}[2]{\ifextended #1\else #2\fi}

\usepackage[pdftex,colorlinks=true,urlcolor=blue,citecolor=black,anchorcolor=black,linkcolor=black]{hyperref}

\makeatletter
\renewcommand{\@seccntformat}[1]{\csname the#1\endcsname\hspace{0.5em}}
\makeatother
\def\subsubsection#1{\savesubsub{#1}}

\begin{document}

% setting up general page style
\pagestyle{fancyplain}

% setting up page style of first page
\thispagestyle{plain}
\firstPageHead{}

% setting up running header (authors) of subsequent pages
\chead{\fancyplain{}{\itshape Alabsi Aljundi, Harrison, Chen, Adiga, Vullikanti, and Marathe}}

% setting up seperation parameters
%\headsep=72pt
\rhead{}
\cfoot{}
\renewcommand{\headrulewidth}{0pt} % (renewcommand needed in fancyhdr to remove top decorative line)
%\headrulewidth=0pt  % ("setlength" needed in fancyheading to remove top decorative line)

%%%%%%%%%%%%%%%%%%%%%%%%%%%%%%%%%%%%%%%%%%%%%%%%%%%%%%%%%%%%%%%%%%%%%%%%%%%%%%
%                                                                            %
%     THESE COMMANDS ARE REQUIRED TO WORK WITH WSC.BST TO MAKE BIBLIO     %
%                                                                            %
%%%%%%%%%%%%%%%%%%%%%%%%%%%%%%%%%%%%%%%%%%%%%%%%%%%%%%%%%%%%%%%%%%%%%%%%%%%%%%
\makeatletter
\let\@internalcite\cite
\def\cite{\def\@citeseppen{-1000}%
    \def\@cite##1##2{(##1\if@tempswa , ##2\fi)}%
    \def\citeauthoryear##1##2##3{##1 ##3}\@internalcite}
\def\citeNP{\def\@citeseppen{-1000}%
    \def\@cite##1##2{##1\if@tempswa , ##2\fi}%
    \def\citeauthoryear##1##2##3{##1 ##3}\@internalcite}
\def\citeN{\def\@citeseppen{-1000}%
%  Pierre L'Ecuyer's fix for multiple cite bug
%  Added by Paul J Sanchez on 4 October 2001
%   \def\@cite##1##2{##1\if@tempswa , ##2)\else{)}\fi}%
%   \def\citeauthoryear##1##2##3{##1 (##3}\@citedata}
    \def\@cite##1##2{##1\if@tempswa, ##2)\else{}\fi}%
    \def\citeauthoryear##1##2##3{##1 (##3)}\@citedata}
\def\citeA{\def\@citeseppen{-1000}%
    \def\@cite##1##2{(##1\if@tempswa , ##2\fi)}%
    \def\citeauthoryear##1##2##3{##1}\@internalcite}
\def\citeANP{\def\@citeseppen{-1000}%
    \def\@cite##1##2{##1\if@tempswa , ##2\fi}%
    \def\citeauthoryear##1##2##3{##1}\@internalcite}
\def\shortcite{\def\@citeseppen{-1000}%
    \def\@cite##1##2{(##1\if@tempswa , ##2\fi)}%
    \def\citeauthoryear##1##2##3{##2 ##3}\@internalcite}
\def\shortciteNP{\def\@citeseppen{-1000}%
    \def\@cite##1##2{##1\if@tempswa , ##2\fi}%
    \def\citeauthoryear##1##2##3{##2 ##3}\@internalcite}
\def\shortciteN{\def\@citeseppen{-1000}%
%  Pierre L'Ecuyer's fix for multiple cite bug
%  Added by Paul J Sanchez on 2 September 2002
%  should have caught this last year...
%   \def\@cite##1##2{##1\if@tempswa , ##2)\else{)}\fi}%
%   \def\citeauthoryear##1##2##3{##2 (##3}\@citedata}
% Shane G. Henderson fix for extra right bracket at end of optional material June 8, 2005
%    \def\@cite##1##2{##1\if@tempswa, ##2)\else{}\fi}%
    \def\@cite##1##2{##1\if@tempswa, ##2\else{}\fi}%
    \def\citeauthoryear##1##2##3{##2 (##3)}\@citedata}
\def\shortciteA{\def\@citeseppen{-1000}%
    \def\@cite##1##2{(##1\if@tempswa , ##2\fi)}%
    \def\citeauthoryear##1##2##3{##2}\@internalcite}
\def\shortciteANP{\def\@citeseppen{-1000}%
    \def\@cite##1##2{##1\if@tempswa , ##2\fi}%
    \def\citeauthoryear##1##2##3{##2}\@internalcite}
\def\citeyear{\def\@citeseppen{-1000}%
    \def\@cite##1##2{(##1\if@tempswa , ##2\fi)}%
    \def\citeauthoryear##1##2##3{##3}\@citedata}
\def\citeyearNP{\def\@citeseppen{-1000}%
    \def\@cite##1##2{##1\if@tempswa , ##2\fi}%
    \def\citeauthoryear##1##2##3{##3}\@citedata}
%
% \@citedata and \@citedatax:
%
% Place commas in-between citations in the same \citeyear, \citeyearNP,
% \citeN, or \shortciteN command.
% Use something like \citeN{ref1,ref2,ref3} and \citeN{ref4} for a list.
%
\def\@citedata{%
    \@ifnextchar [{\@tempswatrue\@citedatax}%
                  {\@tempswafalse\@citedatax[]}%
}

\def\@citedatax[#1]#2{%
\if@filesw\immediate\write\@auxout{\string\citation{#2}}\fi%
  \def\@citea{}\@cite{\@for\@citeb:=#2\do%
    {\@citea\def\@citea{, }\@ifundefined% by Young
       {b@\@citeb}{{\bf ?}%
       \@warning{Citation `\@citeb' on page \thepage \space undefined}}%
{\csname b@\@citeb\endcsname}}}{#1}}%

% don't box citations, separate with ; and a space
% also, make the penalty between citations negative: a good place to break.
%
\def\@citex[#1]#2{%
\if@filesw\immediate\write\@auxout{\string\citation{#2}}\fi%
  \def\@citea{}\@cite{\@for\@citeb:=#2\do%
    {\@citea\def\@citea{; }\@ifundefined% by Young
       {b@\@citeb}{{\bf ?}%
       \@warning{Citation `\@citeb' on page \thepage \space undefined}}%
{\csname b@\@citeb\endcsname}}}{#1}}%

% (from apalike.sty)
% No labels in the bibliography.
%
\def\@biblabel#1{}
\makeatother

%\newlength{\bibhang}
%\setlength{\bibhang}{2em}

% Indent second and subsequent lines of bibliographic entries. Taken
% from openbib.sty: \newblock is set to {}.
% \renewcommand{\refname}{REFERENCES}

\newdimen\bibindent
\bibindent=0.0em
% SEC: was \def\thebibliography#1{\section*{\refname\@mkboth
% SEC: was   {\uppercase{\refname}}{\uppercase{\refname}}}\list
\def\thebibliography#1{\section*{\refname}\list
   {}{\settowidth\labelwidth{[#1]}
   \leftmargin\parindent
   \itemindent -\parindent
   \listparindent \itemindent
   \itemsep 0pt
   \parsep 0pt}
   \def\newblock{}
   \sloppy
   \sfcode`\.=1000\relax}

           % Set up BiBTeX macros

% needed to make the tex document look more like the word counterpart :-(
\setlength{\baselineskip}{12.7pt}

\title{BOUNDARY DEGREE AS A NODE-LEVEL FEATURE FOR EPIDEMIC SCENARIO IDENTIFICATION IN AGENT-BASED CASCADE SIMULATIONS}

\author{\begin{center}Amro Alabsi Aljundi\textsuperscript{1}, Galen Harrison\textsuperscript{1}, Jiangzhuo Chen\textsuperscript{1}, Abhijin Adiga\textsuperscript{1},\\
Anil Kumar Vullikanti\textsuperscript{1}, and Madhav V. Marathe\textsuperscript{1}\\
[11pt]
\textsuperscript{1}Biocomplexity Institute, University of Virginia, Charlottesville, VA, USA\end{center}
}

\maketitle

\vspace{-12pt}

\section*{ABSTRACT}
Characterizing the scenario underlying an epidemic from its disease cascade is
an important task in simulation analytics. We propose boundary degree, the count
of an infected node's contacts in the contact network that were not
infected, as a per-node cascade feature for this task. Through systematic
ablation on realistic contact networks of Tennessee and Virginia, we show
that boundary degree alone improves scenario identification accuracy by 19\%.
Edge labels, whose importance was observed empirically by prior work,
consistently improve accuracy across settings; we provide theoretical
grounding for this observation. These effects are complementary. We prove that
certain epidemic scenarios are indistinguishable without boundary or edge
information. Prior feature-engineering approaches included aggregate boundary
statistics, but not among the top-ranked features; our per-node representation
reveals their importance clearly. Our results suggest
that contact tracing applications should track contacts with non-infected
individuals, not only transmissions.

\section{INTRODUCTION}

High-resolution agent-based models are increasingly used to study complex
epidemic scenarios involving multiple interventions and heterogeneous
populations~\shortcite{hoops2021high,chen2025epihiper}. These
simulations produce detailed cascade graphs (who-infected-who networks)
capturing individual-level transmission events with rich node attributes (e.g.,
age, occupation) and edge attributes (e.g., interaction
type)~\shortcite{verelst2016behavioural}. Analyzing these cascades to
characterize the underlying epidemic scenario is a key task in simulation
analytics, with implications for understanding disease dynamics and evaluating
interventions~\shortcite{bedson2021review}. In practice, cascade data is
obtained through manual and digital contact
tracing~\shortcite{ferretti2020quantifying}. However,
due to limited adoption and privacy constraints, only a fraction of a cascade is
typically observed~\shortcite{blom2021barriers}. This motivates studying
cascades under partial observation settings.

Most approaches to characterizing epidemic scenarios operate on the epidemic
curve (infections per
day)~\shortcite{nsoesie2011prediction,Nsoesie_Leman_Marathe_2014,Augusta_Deardon_Taylor_2019,Pokharel_Deardon_2014}.
This discards the relational structure of cascades. Machine learning approaches
are being increasingly combined with agent-based simulations for epidemiological
tasks~\shortcite{ye2025integrating}. \shortciteN{harrison2023identifying}
advanced the state of the art by formulating the \prob{} problem: classifying
partially observed cascades into their generating scenarios using hand-crafted
aggregated structural features such as labeled path counts and degree
distributions. In prior work~\shortcite{aljundi2023network}, we developed a
graph analytics framework for cascade simulation ensembles using similar
aggregate features. While effective, these aggregate approaches require
substantial domain knowledge to design features, and, critically, they flatten
node- and edge-level signals into global statistics, potentially obscuring
fine-grained patterns that distinguish scenarios.

In this work, we propose the \textit{boundary degree}, the number of an infected
node's contacts in the underlying network that were not infected, as a per-node
cascade feature for the \prob{} problem. Through systematic ablation using a
Graph Neural Network~\shortcite{Jegelka_2022} (GNN) as the evaluation
classifier, we show that boundary degree alone improves scenario identification
accuracy, on average, by 19\% on realistic social contact networks.
Additionally, we provide theoretical grounding for the empirical observation
of~\shortciteN{harrison2023identifying} that edge labels are highly informative:
edge labels consistently improve accuracy across all settings, and their
effect is complementary to boundary degree. This suggests that the boundary
region between infected and non-infected populations carries much of the signal
that distinguishes epidemic scenarios. Although \shortciteN{harrison2023identifying} included aggregate boundary
statistics in their feature set, these were not ranked among their most
important features; the per-node representation we propose reveals their
importance clearly.
This result has a practical implication: contact tracing applications should
incorporate tracking of contacts with non-infected individuals, not only
transmission events.

\noindent \textbf{Summary of contributions.}
\begin{itemize}
  \item \textbf{Boundary degree as a cascade metric.} We propose boundary degree
    as a per-node feature for cascade analysis and demonstrate both
    empirically (+19\% accuracy on realistic networks) and theoretically that
    it captures scenario-discriminating information that aggregate methods miss.

  \item \textbf{Theoretical grounding for feature importance.} We prove that
    certain scenario pairs are indistinguishable without boundary or edge
    information, backing both our findings and the observation of
    \shortciteN{harrison2023identifying} that edge labels matter.

  \item \textbf{Systematic evaluation.} We validate these features through
    ablation experiments on two realistic social contact networks (Tennessee
    and Virginia) and on stochastic block model networks, outperforming the
    feature engineering approach of the state-of-the-art \prob{} algorithm.

  \item \textbf{Robustness and generalizability.} Our approach achieves over 90\%
    accuracy at $T=70$ with only 20\% coverage and transfers across networks
    with minimal performance loss.
\end{itemize}

The source code is available at
\href{https://gitlab.com/AmroAlJundi/cascade_classification_gnn.git}{gitlab.com/AmroAlJundi/cascade\_classification\_gnn}

\section{RELATED WORKS}

\textbf{Scenario identification.} Mathematical and mechanistic models are
effective methods for understanding disease spread
properties~\shortcite{Kermack_McKendrick_1927}. One approach for learning
the model parameters that fit an observed epidemic spread is through supervised
classification of an epidemic spread into a set of hypothesized model parameter
sets~\shortcite{nsoesie2011prediction,Augusta_Deardon_Taylor_2019,Pokharel_Deardon_2014,Nsoesie_Leman_Marathe_2014}. In this approach,
simulations are carried out using these sets of model parameters and then
subsequently used to train a classifier. The classifier is used to classify an
observed cascade. Early work using this approach targeted epidemic
curves~\shortcite{nsoesie2011prediction,Augusta_Deardon_Taylor_2019,Pokharel_Deardon_2014,Nsoesie_Leman_Marathe_2014}.
\shortciteN{nsoesie2011prediction} learn to classify epidemic curves into different
classes using classical machine learning methods such as RF and SVM, as well as
ensemble mechanisms. \shortciteN{Pokharel_Deardon_2014} extend their work by
considering spatially stratified epidemic curve
data. \shortciteN{Nsoesie_Leman_Marathe_2014} propose using Dirichlet processes for
this task, and show that this approach can identify curves that are not part of
the initial hypothesized parameter sets. \shortciteN{Augusta_Deardon_Taylor_2019} use
neural networks for the classification of epidemic curves, both static (MLP) and
sequence-based models. Finally, due to the wide adoption of contact tracing
data, Harrison et al.~\shortcite{harrison2023identifying} proposed porting the scenario
identification problem to cascades rather than epidemic curves. They considered
learning epidemic characteristics in an adversarial setting under full and
partial observation settings. Multiple scenarios differing in disease model
parameters and interventions are constructed in such a way that these scenarios
have very similar infection counts in the initial stages. Features such as node
and edge attributes, structural features of the cascades like size, degrees,
motif counts, etc. are used in the learning process. Various learning algorithms
are applied and evaluated.

\textbf{Boundary degrees as node features.} Given the spread of a disease in a
network, we term the number of contacts that an infected node has with nodes
that are not infected its \textit{boundary degree}. In this work we show
empirically and theoretically that this per-node statistic captures
scenario-discriminating information that other per-node statistics miss.
Quantities on the infected--uninfected interface appear throughout diffusion
modeling, but, to our knowledge, not in the form we use here. In ODE models of spreading dynamics the
interface is a population-level quantity: the effective-degree models of
\shortciteN{lindquistEffectiveDegreeNetwork2011} do track the
susceptible-neighbor count of an infected node, our boundary degree, but only as
an aggregate count of such nodes in a system of ODEs for SIS and SIR spread.
\shortciteN{gleesonHighAccuracyApproximationBinaryState2011} likewise track the
number of \textit{S-I edges} between infected and non-infected nodes, for the
binary-state SIS model. In diffusion source identification,
\shortciteN{prakashSpottingCulpritsEpidemics2012} use the \textit{attack degree},
a per-node count of how many infected neighbors an uninfected node has, to score
and exonerate candidate nodes when localizing an epidemic's sources.
\shortciteN{chengCanCascadesBe2014} use the number of \textit{border edges} of a
cascade, an aggregate count, as a single feature for predicting cascade growth.

Within scenario identification, the only prior use of boundary degree is the
earlier work~\shortcite{harrison2023identifying}, which computed it as an
aggregate histogram over uninfected boundary nodes. We instead use it as a
per-node feature, defined on infected nodes so that it lives on the cascade
nodes the GNN consumes. We show in Section~\ref{sec:learnability} that boundary
degree carries information the cascade structure alone does not: there are
scenario pairs no learner can separate from unlabeled cascades that become
separable once each node carries its boundary degree. Among the per-node
features we test, boundary degree is by far the most discriminative for scenario
identification (Section~\ref{sec:results}).

\textbf{Integrated modeling approaches in epidemiology.}
Learning approaches are being increasingly used in conjunction with agent-based
simulations for various tasks in epidemiological modeling. A recent survey by
\shortciteN{ye2025integrating} reviews such works. Applications include
forecasting, calibration, parameterization, intervention assessment, and
outbreak detection.
\shortciteN{chopra2023differentiable} propose GradABM, a fast
differentiable design for ABMs where agents are represented as tensors and use a
continuous relaxation of the stochastic epidemiological model.

\textbf{GNNs and network propagation.} GNNs have become effective tools for
studying computational problems related to propagation processes.
~\shortciteN{chopra2023differentiable} use a tensorized representation of agent
states for policy evaluation. Another class of cascade representation learning aims at forecasting
the future size of a cascade. Works in this direction
include~\shortcite{Huang_Wang_Zhang_2019} which is a GNN architecture for
learning representations of cascades.

\section{NETWORK REPRESENTATION OF EPIDEMIC CASCADES}\label{sec:preliminaries}
\textbf{Notation.} We use $[k]$ where $k\in \mathbb{N}_{>0}$ to define the
ordered list $[1, 2, \dots, k]$. Matrices are capital roman letters
(e.g.\ $\mathbf{W}$), with $\mathbf{W}_{i,j}$ the entry in row $i$, column $j$. Let~$G=\big(V_G, E_G, \phi_{V_G},
\phi_{E_G}\big)$ be an attributed graph with node set~$V_G$, edge set~$E_G$,
node attributes~$\phi_{V_G}:V_G\rightarrow\reals^{k_v}$, and edge
attributes~$\phi_{E_G}:E_G\rightarrow\reals^{k_e}$. Let~$\gset$ denote a set of
graphs. We use two random graph families for theoretical analysis and empirical
validation: \erdosrenyi~(ER) graphs~\shortcite{erdos1959random}, denoted
$\er(n,p)$, where each pair of $n$ nodes is connected independently with
probability~$p$; and Stochastic Block Model~(SBM)
graphs~\shortcite{HOLLAND1983109}, denoted $\sbm(n,\pi_k,\mathbf{W})$, where
nodes are partitioned into $k$ blocks according to the probabilities in
vector~$\pi_k$ and $\mathbf{W}_{i,j}$ is the edge probability between blocks
$i$ and $j$.  

\textbf{Contact networks.} In the context of computational
epidemiology, a contact network~$G=(V, E)$ represents the interaction structure
of a population. Each node~$v\in V$ corresponds to an individual and is
associated with demographic attributes (e.g., age, gender, occupation). Each
edge~$(u,v)\in E$ represents an interaction between two individuals and is
typically labeled with contextual information such as the activity during which
the interaction occurred (e.g., home, work, school) and its duration. In this
work, we use high-resolution contact networks derived from digital twins of real
populations, constructed by synthesizing multiple data sources including
demographic data, activity patterns, and built-environment
data~\shortcite{chen2025epihiper,hoops2021high}. The resulting graphs span
millions of nodes. We distinguish the underlying contact network from the
cascade graph, which records the interactions on which transmission succeeded.

\textbf{Network propagation model.} The spread of a contagion through a contact
network is simulated as a discrete-time stochastic process in which nodes
transition between health states. We formalize this as a diffusion model
$\diffusion = (\sigma, \theta, \tau)$, where $\sigma$ is the set of states (or
compartments) a node can occupy, $\theta$ is a set of state transition
functions, and $\tau$ is the node state initialization strategy. We use the
\textit{Susceptible}-\textit{Exposed}-\textit{Infected}-\textit{Recovered}
(SEIR) diffusion model, though our method is generalizable to other compartment
models~\shortcite{marathe2013computational}. At each time step~$t$, an
infectious node~$u$ may infect a susceptible neighbor~$v$ along an edge~$e$ with
probability determined by edge-level properties (such as contact duration and
location type) and the disease transmissibility. Following the infection, $v$
transitions to the exposed state ($E$), then to infectious ($I$) after a dwell
time, and eventually to recovered ($R$). State transitions are modeled as a
Gillespie-style stochastic process~\shortcite{chen2025epihiper,hoops2021high}.
For experiments on realistic populations, we use the SEIR parameterization of
\shortciteN{harrison2023identifying}, simulated with the EPIHIPER
engine~\shortcite{chen2025epihiper}. For theoretical analysis and experiments on
random graphs, we use the independent cascade
(IC)~\shortcite{kempe2003maximizing} model, a special case of SEIR with three
states ($S$, $I$, $R$), in which a node infected at time~$t$ is infectious at
$t+1$ and recovers at $t+2$.

\textbf{Interventions.} Interventions modify the network to prevent or mitigate
disease spread. We consider \emph{vaccination} (reducing transmission
probability on edges incident to vaccinated nodes) and \emph{generic social
distancing} (removing non-essential edges), denoting an intervention model
by~$\inter(\theta)$ with parameters~$\theta$.

\textbf{Epidemic scenarios.} As defined by \shortciteN{harrison2023identifying},
an epidemic scenario is a tuple~$\scn(\gset,\diffusion,\inter)$ that specifies a
network family, a diffusion model, and an intervention model. In some cases, we
consider a fixed network, i.e.,~$\gset=\{G\}$, in which case, we will simplify
the notation as follows:~$\scn(G,\diffusion,\inter)$. 

\textbf{Cascade graphs.} Given a scenario~$\scn(\gset,\diffusion,\inter)$, a
single simulation run produces a \emph{who-infected-who} graph called the
cascade graph~\shortcite{newman2003structure,harrison2023identifying}. The
cascade $\cas(V_\cas, E_\cas, \phi_{V_\cas},\phi_{E_\cas})$ has node
set~$V_\cas\subseteq V_G$ containing all infected nodes and edge
set~$E_\cas\subseteq E_G$ containing only the edges on which successful
transmission occurred. Node and edge attributes are inherited from the
underlying contact network. Because the diffusion process is stochastic,
repeated simulation under the same scenario produces an ensemble of cascades.
The set of all possible cascades on graph set~$\gset$ is denoted
by~$\cset_\gset$.

\textbf{Cascade features.}\label{sec:node_rep} Nodes and edges in a cascade
inherit attributes from the underlying contact network. We categorize the
features available for each node into three types:
\begin{enumerate*}[label=(\roman*)]
  \item \textit{inherited features}: node attributes carried over from the
        contact network, such as an individual's age, race, or gender;
  \item \textit{cascade-aggregated features}: features computed from the cascade
        graph itself, such as a node's \textit{cascade degree} (number of edges
        incident on the node in the cascade); and
  \item \textit{contact network-aggregated features}: features that combine
        information from both the cascade and the underlying contact network,
        such as a node's \textit{graph degree} (total number of edges incident
        to the node in the contact network) and its \textit{boundary degree}
        (defined below).
\end{enumerate*} Edge labels follow a similar taxonomy. In our experiments, we use activity
labels inherited from the contact network (e.g., home, work, school),
represented as one-hot encodings.

\textbf{Boundary degree.} For a node $u$ in cascade $C=(V_C, E_C)$ on contact
network $G=(V_G, E_G)$, the boundary degree is defined as: \[\beta_u = |\{v :
v\notin V_C,\ (u,v) \in E_G\}|\] That is, $\beta_u$ counts the number of $u$'s
contacts in the underlying network that were \textit{not} infected. In this
work, we demonstrate theoretically and empirically the high signal embedded in
this feature.

\textbf{Partial observation model.}  We model partial information of cascades
using the following model from \shortciteN{harrison2023identifying}.
For~$0<\kappa\le1$, let~$V_G'\subseteq V_G$ be a set of~$\kappa\cdot n$ nodes
sampled uniformly. Given a cascade~$\cas$, $\obssamp(\cas,\kappa)$ is the
subgraph of~$\cas$ induced on the node set~$V_G'\cap V_\cas$.

We formally define the scenario identification problem as defined in
\shortciteN{harrison2023identifying}. 
\begin{problem}[\prob]
  Suppose we are given a set of epidemic scenarios $\scnset = \{\scn_1, \scn_2,
  \cdots, \scn_k\}$ defined on graph set $\gset$, a partial observation
  model~$\obs(\cdot)$, and a collection of labeled (partially observed)
  attributed cascades~$\dataset=\{(\obs(\cas_i),\ell_i)\mid i\in[N],\ \cas_i\in
  \cset_\gset,\ \ell_i\in[k]\}$, where~$\ell_i=j$ implies that~$\cas_i$ was
  generated by scenario~$j$. The objective is to find a
  function~$f:\cset_\gset\rightarrow [k]$ that, given an unlabeled
  cascade~$\cas\in\cset_\gset$, predicts the scenario in $\scnset$ that
  generated it.
\end{problem}

\begin{figure*}
  \centering \includegraphics[width=0.7\textwidth]{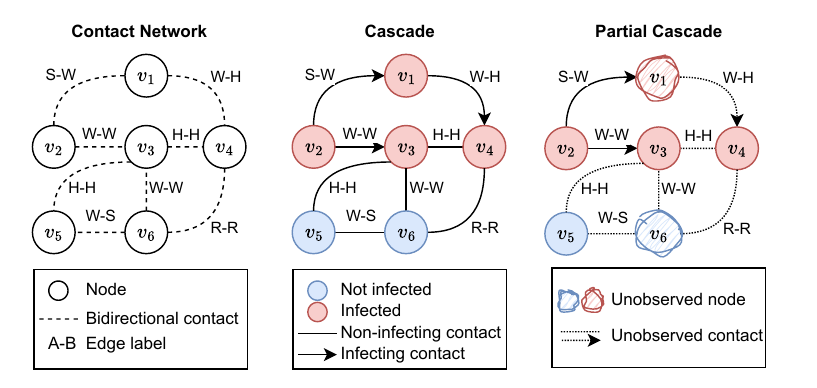}
  \caption{A contact network (left), a cascade generated over it (center),
    and the same cascade under partial observation with $\kappa\approx 0.3$
    (right). Partial observation omits nodes and their edges from the contact
    network. Node $v_3$ has boundary $\beta_3 = 2$ under the fully observed
    cascade (center).
    Under partial observation (right), $v_6$ is unobserved while $v_5$ remains
    observed, but the edge $(v_3, v_5)$ is unobserved, so $\beta_3 =
    0$.}\label{fig:example}
\end{figure*}
Figure~\ref{fig:example} illustrates a contact network and a cascade generated
over it.

\section{EXPERIMENTAL DESIGN} \label{sec:eval} We conducted two sets of
experiments, one on realistic contact networks of human population and the other
one on a set of stochastic block model networks. We describe their experimental
setup below.

\paragraph{Synthetic population datasets (\popdata)} We use realistic contact
networks of two US states, Tennessee~(TN) and
Virginia~(VA)~\shortcite{virginiadata}, which have been applied in multiple
studies~\shortcite{chen2025epihiper,harrison2023identifying,hoops2021high}.
This network data is derived from a US-scale digital twin, which in turn is
obtained by fusing diverse data sets such as census data (providing nodes with
attributes), land use data, activity patterns, building maps, etc. We construct
two datasets of cascade graphs generated through SEIR disease spread simulations
on the synthetic contact graphs. Structural properties of these networks appear
in Table~\ref{tab:graphs}. In contact graph~$G(V,E)$, nodes are individuals and
edges are interactions between individuals. Nodes have numerical (age) and
categorical (gender, race, designation, age group) attributes. Every edge is
attributed with two labels that represent the activity that the nodes on either
end of the edge were engaged in at the time of the interaction. Activities
include: ``home'', ``work'', ``shop'', ``school'', ``college'', ``religion'',
and ``other''. Mathematically, the label of edge $(u,v)$ is the concatenation of
two one-hot-encoded vectors of length 7, each representing the activity of nodes
$u$ and $v$ at the time of interaction.

\begin{table}[h!] \caption{Structural information about the synthetic contact
  networks from Harrison et al. used here. Both networks have a diameter of
length 14.} \label{tab:graphs} \centering \begin{tabular}{c|cccc} Network & |V|
  & |E|        & Max degree & Average degree \\ \hline TN      & 6,041,517 &
62,149,441 & 461    & 20.6   \\ VA      & 7,602,717 & 83,162,927 & 543    & 21.9
\end{tabular} \end{table}

We use four scenarios to generate cascades as defined in
\shortcite{harrison2023identifying}. Scenarios vary by disease transmissibility
and two intervention parameters. The vaccination intervention
$\inter_{vax}(\sigma_{vax}, \alpha)$ reduces the probability of infection of
$|V|\cdot\sigma_{vax}$ randomly selected nodes by $\alpha = 80\%$. The generic
social distancing intervention $\inter_{gsd}(\sigma_{gsd})$ removes, for
$|V|\cdot\sigma_{gsd}$ randomly selected (compliant) nodes, all incident edges
in which either endpoint is engaged in a non-essential activity (e.g.\
``shop''). Parameters were chosen to ensure that the cascades generated under
different scenarios are not easily distinguishable. Per-scenario parameters
appear in Table~\ref{tab:design}.

\begin{table}[h!] \caption{Disease and intervention parameters for different
  scenarios. $^*$Scenarios on the VA network used $\tau = 0.155$.
$^\dag$Scenario on the VA network used $\sigma_{gsd} = 60\%$.}
\label{tab:design} \centering \begin{tabular}{ c|ccc } \textbf{Scenario} &
  $\tau$ & $\sigma_{vax}$ & $\sigma_{gsd}$\\ \hline No vax + Low GSD & 0.09 &
  None & 25\%\\ No vax + High GSD & 0.09 & None & 70\%\\ Vax + Low GSD &
  $0.16^*$ & 50\% & 25\%\\ Vax + High GSD & $0.16^*$ & 50\% & $65\%^\dag$\\
\end{tabular} \smallskip \end{table}

To generate a cascade, the simulator applies interventions, randomly infects 20
nodes, and allows the disease to propagate for 300 days. For each contact
network, we generate 100 cascades per scenario. We create three datasets for
each contact network, each with a different time horizon (i.e. a time point at
which infections are no longer recorded). We use time horizons
$T\in\{30,50,70\}$ as it was shown that for this setting, scenarios are less
distinguishable earlier in the simulation~\shortcite{harrison2023identifying}.
Cascades are generated using the EPIHIPER epidemic
simulator~\shortcite{chen2025epihiper}. \paragraph{Stochastic Block Model (SBM)
datasets.}\label{sec:eval:sbm} We also experiment on SBM random graphs with
community-labeled edges and IC diffusion. Nodes are attributed with boundary,
cascade, and graph degrees as defined in Section~\ref{sec:node_rep}. An edge is
attributed by two concatenated vectors, each being the one-hot encoding of the
communities of the nodes at its ends.  We define a scenario $\scn=(\sbm(n,
\pi_k, \mathbf{W}), \icmodel_{\sbm}(\mathbf{T}, \nu))$  by its SBM parameters as
described in Section~\ref{sec:preliminaries} as well as its diffusion model
parameters $\mathbf{T} \in [0,1]^{k \times k}$ and $\nu$ which, combined, define
transition probabilities. To elaborate, given these parameters, an infected node
$u$ in community $i$ infects a susceptible node $v$ in community $j$ with
probability $\mathbf{T}_{i,j} \cdot \nu$. We use two datasets:
\begin{enumerate}
  \item \textbf{\sbmstudy{}} is a set of 1128 scenario pairs that represent a
        wide range of cross-scenario structural variety. All scenario pairs have
        $n=2000$, equisized communities, and symmetric $\mathbf{W}$ and
        $\mathbf{T}$ matrices. We sweep across the following parameters
        $\mathbf{W}_{i,j\neq i} \in \{0.005, 0.001\}$, $\mathbf{W}_{i,i} \in
        \{0.031, 0.053, 0.073\}$, $\mathbf{T}_{i, j\neq i}\in \{0.5, 0.3,
        0.1\}$, $\mathbf{T}_{i, i} = 1$, and $\nu \in \{0.029, 0.047, 0.071\}$.
        For each scenario pair, we run 1000 simulations. We prune scenario pairs
        based on two criteria: \begin{enumerate*}
          \item the number of infected nodes is less than 5\% of the total nodes
            in the graph, and
        \item one scenario's mean infection count exceeds $6\times$ the
            other's.
        \end{enumerate*}

      \item \textbf{\sbmthreecomm{}}, is a set of scenario pairs designed so
        that cascades are indistinguishable without community membership
        information. Each cascade contains three communities, and for scenario
        pair $(\scn_1, \scn_2)$ we guarantee that $\pr(\cas \in \scn_1) =
        \pr(\cas' \in \scn_2)$, where $\cas'$ is equivalent to $\cas$ but with
        the first and second communities' nodes swapping communities. Scenario
        pairs vary across the separability parameter $s\in [0.5, 1]$ which
        controls how similar inter-community transmission probabilities are
        between a scenario pair; when $s=0.5$ both $\scn_1$ and $\scn_2$ have
        identical transmission matrices; when $s=1.0$, the transmission
        probabilities $1 \leftrightarrow 3$ and $2 \leftrightarrow 3$ are the
        exact opposite across the scenario pair. All scenarios use $n=3000$
        nodes in three equally sized communities and share the symmetric SBM
        matrix $\mathbf{W}$ below; given global transmissibility
        $\nu=1$, $s\in[0.5,1]$, and $\gamma=1-s$, the transition matrices
        $\mathbf{T_1},\mathbf{T_2}$ and $\mathbf{W}$ are

        \begin{equation*}
          \mathbf{T_1}=\begin{bmatrix}1&0.5\gamma&0.5s\\0.5\gamma&1&0.5\gamma\\0.5s&0.5\gamma&1\end{bmatrix},\;
          \mathbf{T_2}=\begin{bmatrix}1&0.5\gamma&0.5\gamma\\0.5\gamma&1&0.5s\\0.5\gamma&0.5s&1\end{bmatrix},\;
          \mathbf{W}=\begin{bmatrix}0.031&0.005&0.005\\0.005&0.031&0.005\\0.005&0.005&0.031\end{bmatrix},
        \end{equation*} where the off-diagonal (inter-community) entries are
        scaled by $0.5$. 
    \end{enumerate}

Further details of the SBM graph generation, diffusion parameters, and data
representation are given in \apptext{Appendix~\ref{app:eval:sbm}}{Appendix C.2
in the extended version~\shortcite{aljundi2026boundarydegreenodelevelfeature}}.

\paragraph{Graph classifier} To evaluate the importance of individual cascade
features through ablation, we use a Graph Neural Network (GNN) as the
classification model~\shortcite{Jegelka_2022}, implemented using the GraphGym
framework~\shortcite{you2020design}. The GNN follows the Message Passing Neural
Network (MPNN) formulation~\shortcite{mpnn}: over $L$ layers,
each node aggregates information from its neighbors and their incident edges to
update its hidden representation. After the final layer, node representations
are pooled to produce a single cascade-level vector, which is classified by a
feed-forward head. The GNN takes cascade features (as defined in
Section~\ref{sec:node_rep}) as input without requiring hand-crafted aggregation
into global statistics. We use a standard MPNN architecture without
architectural optimizations to isolate the effect of input features from model
design choices. Architectural details and hyperparameters are given in
\apptext{Appendix~\ref{app:gnn}}{Appendix B in the extended
version~\shortcite{aljundi2026boundarydegreenodelevelfeature}}. 

\paragraph{Baseline} We use the model proposed by
\shortciteN{harrison2023identifying} (\epicurveandstructure{}) as our baseline.
It uses hand-crafted structural features of cascades (labeled path counts,
degree distributions) classified with logistic regression, random forest, and
SVM.

\section{RESULTS}\label{sec:results}
We use classification accuracy to evaluate which cascade features are most
informative for scenario identification. We first identify boundary degree and
edge labels as the dominant signals through systematic ablation. We then
validate these findings by comparing against a feature engineering baseline, and
test robustness under partial observation and generalizability across networks.

We use a 75\%/25\% train-evaluation split; since the classes are balanced, we
report accuracy, averaged over 5 runs with different random splits.

\subsection{Feature Importance Analysis}
A key advantage of GNNs is that they operate directly on per-node and per-edge
features, giving us the ability to evaluate individual feature contributions
through ablation. We systematically test which node features and edge labels
carry the most discriminative signal for scenario identification.

\subsubsection{Boundary Degree}
We measure the contribution of the aggregated degree features (graph degree,
cascade degree, and boundary degree) by training the GNN with and without each
feature. Figure~\ref{fig:edge_counts}.a shows the per-feature ablation on
\popdata{}; Figure~\ref{fig:sbm}.b shows boundary degree's effect on \sbmstudy{}.

\begin{figure*}[t]
    \centering
    \begin{minipage}[t]{0.48\textwidth}
        \centering
        \includegraphics[width=\linewidth]{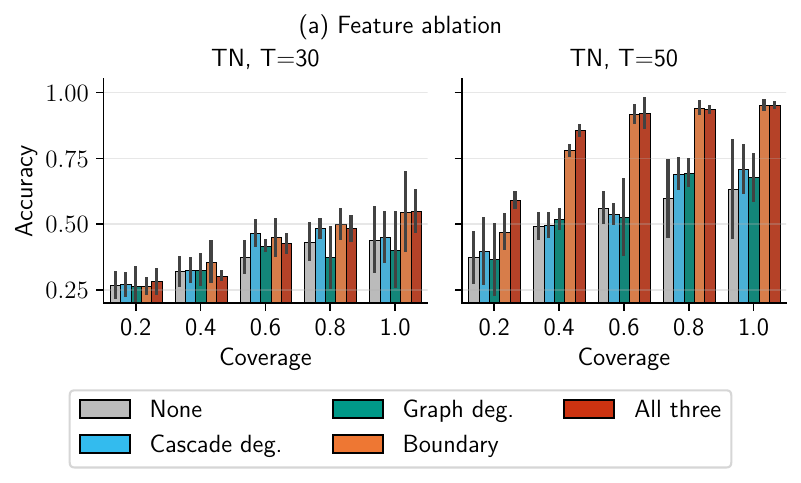}
    \end{minipage}%
    \hfill
    \begin{minipage}[t]{0.48\textwidth}
        \centering
        \includegraphics[width=\linewidth]{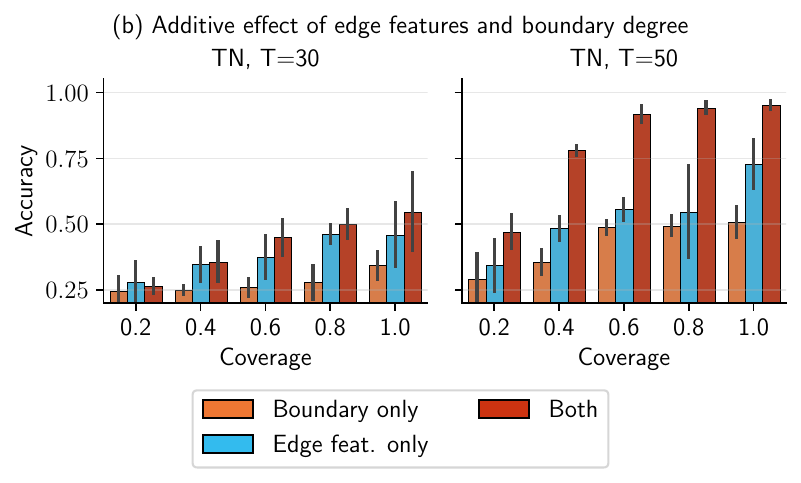}
    \end{minipage}
    \caption{Feature importance analysis on TN at $T \in \{30, 50\}$.
	    (a) GNN accuracy with each aggregated node feature added individually.
    (b) complementarity of boundary degree and edge labels.}
    \label{fig:edge_counts}
\end{figure*}

\begin{figure*}[b]
    \centering \includegraphics[width=0.72\textwidth]{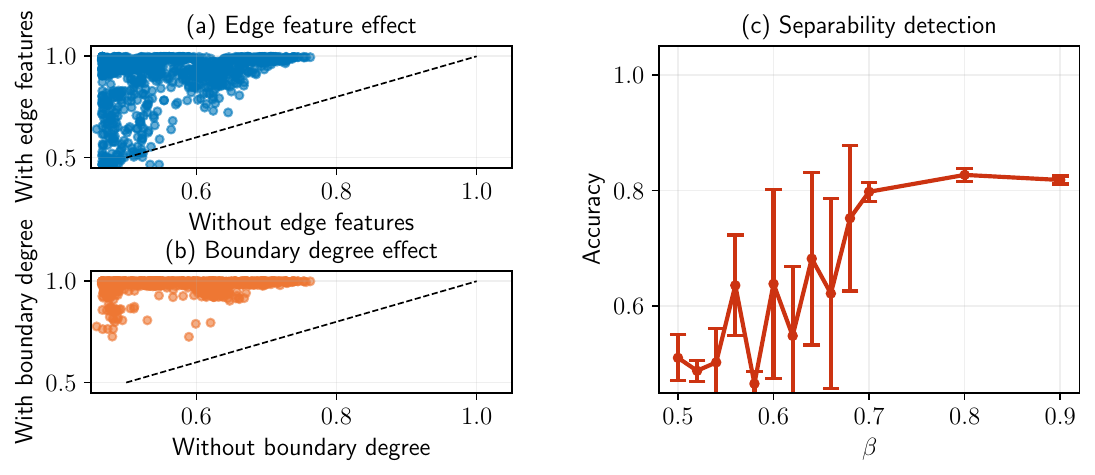}
    \caption{(a)~Effect of edge labels on \sbmstudy{}: each
    point is a scenario pair, with accuracy without edge labels on the
    $x$-axis and with edge labels on the $y$-axis. (b)~Effect of boundary
    degree on the same dataset. In both (a) and (b), points above the dashed
    $x=y$ line indicate an improvement. (c)~Separability
    detection on \sbmthreecomm{}: accuracy increases as the separability
    parameter $s$ grows, confirming that the GNN can exploit edge labels to
    distinguish scenarios that are provably indistinguishable without them.}
    \label{fig:sbm}
\end{figure*}

The results are clear: boundary degree is the dominant feature. It alone
improves accuracy by an average of 17\% on TN and 22\% on VA across all time
horizons and coverages (19\% overall, Wilcoxon signed-rank $p < 0.001$,
improvement in 151 out of 160 conditions). Adding cascade and graph degrees on
top contributes only 1--2\% additional accuracy. On \sbmstudy{}, boundary degree
improves accuracy, on average, by 42\%. We observe a similar effect on 
alternative GNN architectures including one with stochastic edge
sampling~\shortcite{you2020design}, attention
mechanisms~\shortcite{brody2021attentive}, and neighborhood
sampling~\shortcite{hamilton2017inductive}, confirming that the boundary degree
effect is not architecture-specific.

This finding is significant because \shortciteN{harrison2023identifying}
included aggregate boundary statistics (binned boundary degree histograms) in
their feature set, but their SHAP analysis found labeled path counts to be the
dominant group, with boundary statistics not among the top-ranked features. The
GNN, which processes boundary degree at the per-node level distributed across
the cascade, reveals its importance clearly. This confirms our theoretical
result at the end of Section~\ref{sec:learnability}: boundary degree captures
information that is otherwise invisible to methods operating on unlabeled
cascades.

\subsubsection{Edge Labels}
Throughout our experiments, edge labels proved to be of significant importance
for \prob{}. This result corroborates the results
of~\shortcite{harrison2023identifying} where labeled path counts were among the
most important features for classification quality. Edge labels consistently
improve accuracy across all settings. On the \sbmstudy{} dataset, edge labels
improve accuracy for every one of the 1128 scenario pairs
(Figure~\ref{fig:sbm}.a). We also validate edge label importance on
\sbmthreecomm{}, where scenario pairs are designed so that cascades are
indistinguishable without community membership information encoded in edge
labels. Recall from Section~\ref{sec:eval} that the separability parameter
$s\in[0.5,1]$ controls how differently the two scenarios connect community 3 to
the rest of the graph: at $s=0.5$ the two scenarios are identical, and at $s=1$
they are fully separable. We find that, as $s$ increases, accuracy increases
(Figure~\ref{fig:sbm}.c). Finally, we find the effect of edge labels and
boundary degrees to be complementary, as shown in
Figure~\ref{fig:edge_counts}.b.
\subsection{Classification Performance}
Having established which features matter, we now validate that a GNN using these
features outperforms the feature engineering baseline of Harrison et
al.~\shortcite{harrison2023identifying}. 
\begin{figure}
    \centering \includegraphics[width=0.9\linewidth]{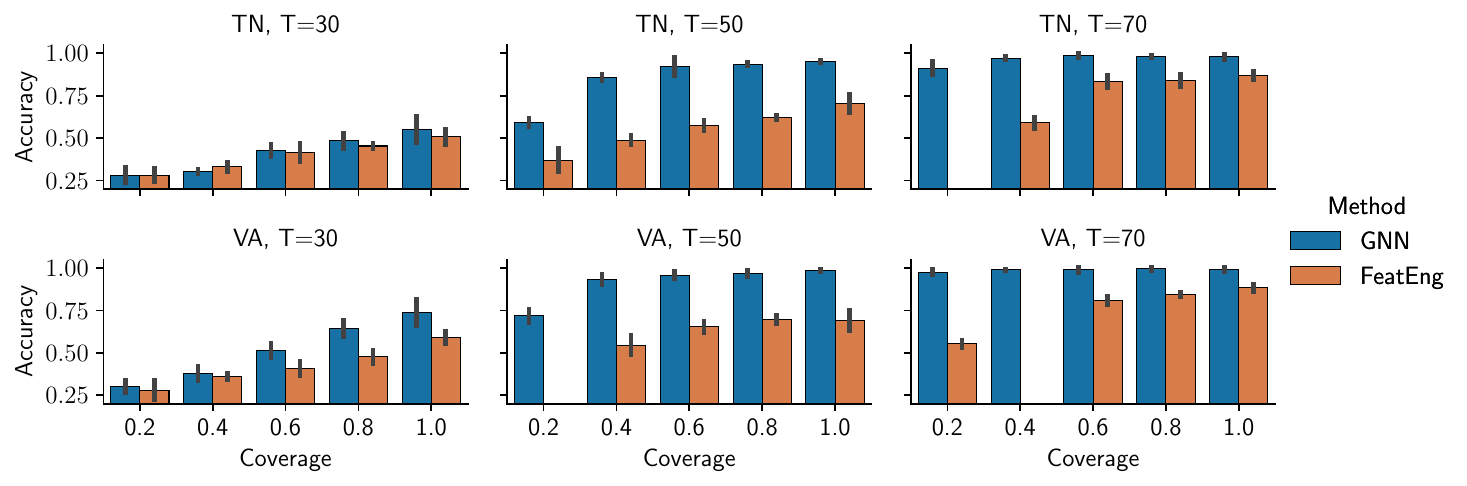}
    \caption{Comparing the classification accuracy on \popdata{} over $T\in\{30,
    50, 70\}$ and $\kappa\in\{0.2, 0.4, 0.6, 0.8, 1.0\}$ using the GNN model and
    \epicurveandstructure{}.} \label{fig:gnn_vs_es}
\end{figure}

Figure~\ref{fig:gnn_vs_es} shows the GNN and \epicurveandstructure{} at
different time horizons and coverage percentages. The GNN outperforms
\epicurveandstructure{} in 25 out of 27 matched conditions (Wilcoxon
signed-rank, $p < 0.001$), with a mean accuracy improvement of 17.6 percentage
points. At $T=50$ and $T=70$, the GNN under $40\%$ and $20\%$ coverage
outperforms \epicurveandstructure{} at full coverage. The results also
demonstrate robustness to partial observation: the GNN classifies scenarios at
$T=70$ with more than $90\%$ accuracy using only $20\%$ of the data, and
outperforms \epicurveandstructure{} at every coverage level even at $T=30$.

\subsection{Generalizability Across Networks}
We evaluate whether the learned representations transfer across contact
networks. A model trained on TN and evaluated on VA (with matching time horizon
and coverage) is, on average, only 6\%, 5\%, and 1\% less accurate than a model
trained directly on VA at times 30, 50, and 70, respectively. Additional
transfer learning experiments and fine-tuning results can be found in
\apptext{Appendix~\ref{app:eval:transfer}}{Appendix D.1 of~\shortcite{aljundi2026boundarydegreenodelevelfeature}}.

\section{LIMITS ON LEARNABILITY}\label{sec:learnability}
We develop theoretical limits on learnability. We exhibit pairs of scenarios
that cannot be distinguished from their cascades alone, under uncertainty in
either the diffusion parameters or the underlying network, and show how node
features such as boundary degree restore distinguishability.

Recall that a scenario~$\scn(\gset,\diffusion,\inter)$ is defined by the graph
class~$\gset$, diffusion model~$\diffusion$, and an intervention model~$\inter$.
Depending on the setting, we will consider a fixed graph~$G$ or a class of
graphs~$\gset$.  All the results in this section correspond to the IC model~(see
Section~\ref{sec:preliminaries} for definition). Let~$G$ be the underlying
graph. We consider a homogeneous setting where the transmission probability on
all edges of the graph is~$\pt$. We will assume that every cascade is rooted at
node~$r$~(the initial condition for each scenario), which is known to the
observer. Let $\diffusion=\ic(\pt,r)$ denote the IC diffusion model with
transmission probability~$\pt$ seeded at node~$r$. We consider edge~($\inter_e$)
and node interventions~($\inter_v$). Note that for a cascade~$\cas$ to occur,
transmission has to occur exactly on the edges of the cascade~$E_C$.  This
implies that other edges should have failed to transmit. Let~$\dc=|\{(u,v)\in
E_G\setminus E_C\mid u\in V_C,~v\in V_G\}|$ denote the number of
\emph{failed-transmission} edges: edges of the contact network with at least one
endpoint in the cascade~$\cas$ that are not part of~$E_C$, i.e., contacts of
infected nodes on which transmission did not occur. The probability of
generating a cascade~$\cas$ with specific node set~$V_C$ (containing root~$r$)
and edge set~$E_C$ with~$\mc$ edges is as follows:
\begin{align}
\label{eqn:prob_cascade} \pr(\cas)=\pt^\mc (1-\pt)^{\dc}\,.
\end{align}
Note: To obtain Equation~\ref{eqn:prob_cascade}, we are relaxing the IC model
constraint that an infected node can only infect susceptible nodes. We allow for
a node infected at time $t$ to be infected by a neighbor infected at the same time
step. But regardless of whether this event happens or not, this node will
recover at time step $t + 1$. This allows for transmission between nodes which are
at the same distance from the root node in the cascade, but does not alter the
evolution of the system state in any way.

First, we consider a simple setting where the graph~$G$ is fixed and an
edge-based quarantining is in place: two end points of any edge~$(u,v)$ decide
to not interact with one another with probability~$p_i$ independent of other
edges. Let this intervention model be denoted as~$\inter_e(p_i)$.

\begin{proposition}
\label{thm:edge_intervention} Given two scenarios
$\scn_1\big(G,\ic(\pt,r),\inter_e(p_i)\big)$ and
$\scn_2\big(G,\ic(\pt',r),\inter_e(p_i')\big)$ such that~$\pt\cdot
(1-p_i)=\pt'\cdot (1-p_i')$, and any set of training
samples~$\big\{(\cas,\ell)\mid \cas\in\cset,~\ell\in\{1,2\}\big\}$, no learner
can solve the $\prob(\scn_1,\scn_2)$.
\end{proposition}
A similar argument handles network uncertainty, where $G\in\gset(n,p_e)$ is an
\erdosrenyi{} graph with edge probability~$p_e$.

\begin{proposition}
\label{thm:network_uncertainty}
Given two scenarios $\scn_1\big(\gset(n,p_e),\ic(\pt,r)\big)$ and
$\scn_2\big(\gset(n,p_e'),\ic(\pt',r)\big)$ such that~$\pt\cdot p_e=\pt'\cdot
p_e'$, and any set of training samples~$\big\{(\cas,\ell)\mid
\cas\in\cset,~\ell\in\{1,2\}\big\}$, no learner can solve the
$\prob(\scn_1,\scn_2)$.
\end{proposition}

Proofs of propositions~\ref{thm:edge_intervention}
and~\ref{thm:network_uncertainty} can be found in
\apptext{Appendix~\ref{app:learnability}}{Appendix A in the extended
version~\shortcite{aljundi2026boundarydegreenodelevelfeature}}. 

\begin{remark}
For simplicity, we have considered homogeneous edge probabilities in the above
results. They can be easily extended to the heterogeneous setting (set the
product of the probabilities to~$\beta_e$ for each edge~$e$ instead of the
same~$\beta$).
\end{remark}

Now, we show that there exist pairs of scenarios defined on graphs with
attributes or node~(or edge) labels that cannot be distinguished from one
another if observed cascades are unlabeled. But these pairs can be easily
identified if the cascade nodes and edges inherit the labels from the graph. For
a graph~$G$ with node labels, let~$\cset$ denote the set of cascades without
node labels and~$\cset^+$ denote the set of cascades with node labels. 
\begin{proposition}
\label{thm:labeled}
There exists a graph~$G$ with node labels and two scenarios~$\scn_1$
and~$\scn_2$ such that (i)~for any set of training
samples~$\big\{(\cas,\ell)\mid \cas\in\cset,~\ell\in\{1,2\}\big\}$, no learner
can solve $\prob(\scn_1,\scn_2)$; and (ii)~there exists a GNN that can solve
$\prob(\scn_1,\scn_2)$ given training samples consisting of labeled
cascades~$\big\{(\cas^+,\ell)\mid \cas^+\in\cset,~\ell\in\{1,2\}\big\}$.
\end{proposition}

\begin{proof}[Proof sketch]
We build a graph with three node clusters $V_1,V_2,V_3$ where $G(V_1\cup V_3)$
and $G(V_2\cup V_3)$ are isomorphic; the two scenarios are seeded in $V_1$ and
$V_2$ with their cross-cluster edges removed. By this isomorphism the scenarios
induce identical distributions over \emph{unlabeled} cascades, so no learner
separates them. If each node instead carries its cluster as a one-hot feature, a
single GNN layer that \emph{sums} the per-node features yields a cascade-level
aggregate that differs across scenarios, which a perceptron separates
(Perceptron Convergence theorem~\shortcite{rosenblatt1958perceptron}). The full
construction is in
\apptext{Appendix~\ref{app:learnability}}{Appendix~A of the extended
version~\shortcite{aljundi2026boundarydegreenodelevelfeature}}.
\end{proof}

\paragraph{Effect of including boundary degree as node feature.} Finally, we
show the effectiveness of including local structural features of nodes by
revisiting Proposition~\ref{thm:network_uncertainty}. We recall that no learner
can distinguish between the two scenarios from unlabeled cascades in this case.
However, including each node's boundary degree as a feature makes the two
scenarios distinguishable by the same sum-then-perceptron learner used in the
sketch above. For~$\scn_1$ and~$\scn_2$, the expected boundary
degrees of any node are~$\beta=p_e(1-\pt)n$ and~$\beta'=p_e'(1-\pt')n$
respectively. One can choose the edge probabilities and transmission
probabilities such that~$p_e\pt=p_e'\pt'$ but the difference between~$\beta$
and~$\beta'$ to be very large. Clearly, this learner separates the scenarios
from the~$\beta$ values alone.

\section{CONCLUSION}
We propose boundary degree as a per-node cascade feature for epidemic scenario
identification. Through ablation on realistic and SBM networks with an MPNN
classifier, we show that boundary degree alone improves accuracy by 19\%, that
edge labels (whose importance \shortciteN{harrison2023identifying} observed
empirically) consistently help, and that the two are complementary. We prove
that certain scenario pairs are indistinguishable without boundary or edge
information, giving theoretical grounding for both.

The prior feature-engineering approach~\shortcite{harrison2023identifying}
included aggregate boundary statistics but did not rank them among its top
features; representing boundary degree per node makes its importance clear. We
therefore recommend that contact tracing applications track contacts with
non-infected individuals, not only transmissions. The classifier is robust to
partial observation, reaching over 90\% accuracy at $T=70$ under 20\% coverage,
and generalizes across networks with minimal loss.

Our study has several limitations. The learnability results assume homogeneous
transmission and the IC model. Boundary degree requires knowledge of the underlying
contact network, which motivates our contact-tracing recommendation but may be
only partially available in practice. Finally, our empirical evaluation uses
synthetic populations of two US states and a fixed set of scenarios; broader
populations and scenario families remain to be tested.

In the future, we envision using GNNs more broadly to analyze agent-based
epidemic simulations, with scenario identification being one task among several.
The same approach applies to forecasting, source detection, and assessing
interventions, with the GNN operating on the cascade directly, not hand-crafted
summaries. Which architectures capture cascade
structure and how choice of features interacts with the architecture are
important aspects to consider. Instead of classifying cascades into a fixed set of scenarios,
another direction is to estimate the underlying disease and intervention
parameters directly, and to detect cascades that come from scenarios we did not
train on. Finally, extending our learnability results beyond the homogeneous IC model to
heterogeneous and SEIR dynamics and validating the approach on real cascade
data is a natural next step.

\section*{ACKNOWLEDGMENTS}
%% This work was partially supported by the University of Virginia Strategic Investment Fund (SIF160); NSF grants CCF-1918656, CCF-1917819, OAC-2027541, IIS-2027007, IIS-1633028, and OAC-1916805;  DTRA subcontract/ARA S-D00189-15-TO-01-UVA, and the
%% % PGCoE UVA-BI
%% Centers for Disease Control and Prevention (CDC) through Pathogen Genomics Centers of Excellence network (PGCoE) grant 6NU50CK000555-03-01.

This work was partially supported by the University of Virginia Strategic Investment Fund SIF160 and SIF176A Contagion
Science; NSF grants CCF-1918656, CCF-1917819, OAC-2027541, IIS-2027007, IIS-1633028, and OAC-1916805;  DTRA subcontract/ARA S-D00189-15-TO-01-UVA, and the
% PGCoE UVA-BI
Centers for Disease Control and Prevention (CDC) through Pathogen Genomics Centers of Excellence network (PGCoE) grant 6NU50CK000555-03-01.

\ifextended
\appendix
\clearpage
\appendix
\section{ADDITIONAL DETAILS FOR SECTION~\ref{sec:learnability} LEARNABILITY}\label{app:learnability}

\begin{proof}[\bf Proof of Proposition~\ref{thm:edge_intervention}]
In~$\scn_1$, the probability that an infected node~$u$ infects neighbor~$v$
is~$(1-p_i)\cdot\pt$; the two nodes interact with one another and then
transmission occurs. Therefore, for any cascade~$\cas$, by
(\ref{eqn:prob_cascade}), $\pr(\cas)=\beta^\mc (1-\beta)^\dc$
where~$\beta=(1-p_i)\cdot\pt$. By the same argument as above, the probability
that an infected node~$u$ infects neighbor~$v$ in $\scn_2$
is~$(1-p_i')\cdot\pt'=\beta$. Hence, every cascade~$\cas$ has the same
probability of occurrence in both scenarios or in other words,~$\cset$ has
identical distributions in both scenarios, i.e.,~$\Pr(\scn_1\mid\cas) =
\Pr(\scn_2\mid\cas)=\frac{1}{2}$. Hence, it is impossible to determine the
scenario by only observing the cascades.
\end{proof}

\begin{proof}[\bf Proof of Proposition~\ref{thm:network_uncertainty}]
In~$\scn_1$, for any cascade~$C(V_\cas,E_\cas)$, the probability of~$e\in
E_\cas$ is~$p_e\cdot\pt=\beta$. For any~$u\in V_\cas$ and $(u,v)\notin E_\cas$
the probability that transmission does not occur on this edge is~$1-\beta$. The
number of such~$(u,v)$ candidates is~$\nc\cdot n-\mc$. Hence,
following~(\ref{eqn:prob_cascade}),~$\pr(\cas)=\beta^\mc(1-\beta)^{\nc\cdot
n-\mc}$. As in Proposition~\ref{thm:edge_intervention},~$\cas$ has the same
probability in~$\scn_2$ as well. Hence, proved.
\end{proof}
\begin{proof}[\bf Proof of Proposition~\ref{thm:labeled}]
We will first construct~$G$ and the two scenarios. The graph~$G$ consists of
a node set that can be partitioned into three clusters:~$V=V_1\uplus V_2\uplus
V_3$. For every node~$v_{1i}\in V_1$, the corresponding node in~$V_2$
is~$v_{2i}$. Let~$\pi: V_1\cup V_3\rightarrow V_2\cup V_3$ be defined as
follows:~$\forall~v_{1i}\in V_1$,~$\pi(v_{1i})=v_{2i}$ and~$\forall~v\in
V_3$,~$\pi(v)=v$. There exist no edges between~$V_1$ and~$V_2$.  Also, for every
pair~$u\in V_1, v\in V_3$, $\{u,v\}\in E(G)\Leftrightarrow \{\pi(u),\pi(v)\}\in
E(G)$.  The graph induced on~$V'\subset V$ is denoted by~$G(V')$.  Note
that~$G(V_1)$ is isomorphic to~$G(V_2)$. Also,~$G(V_1\cup V_3)$ is isomorphic
to~$G(V_2\cup V_3)$. Both scenarios have the same diffusion model but with
different root nodes~$r_\ell\in V_\ell$,~$\ell=1,2$. In
scenario~$\scn_\ell(G,\ic(\pt,r_\ell),\inter_\ell)$,~$\ell=1,2$, we implement
the edge intervention scenario~$\inter_\ell$ where all edges between~$V_\ell$
and~$V_{3-\ell}$ are removed.

Note that in~$\scn_\ell$, only nodes from~$V_\ell$ and~$V_3$ can be infected. By
isomorphism of~$G(V_1\cup V_3)$ and~$G(V_2\cup V_3)$, for every cascade~$C_1$
from~$\scn_1$, there exists a corresponding unique cascade~$C_2$ corresponding
to~$\scn_2$ such that~$\Pr(C_1)=\Pr(C_2)$. The two cascades have the same nodes
from~$V_3$, and satisfy (i)~$v\in V(C_1)\Leftrightarrow \pi(v)\in V(C_2)$ and
(ii)~$\{u,v\}\in E(C_1)\Leftrightarrow \{\pi(u),\pi(v)\}$. By the same argument
as in the previous result, no learning algorithm can distinguish between the two
scenarios.

Suppose each node was labeled based on its cluster membership using one-hot
encoding. For a node~$v$, let~$h(v)=[I_1(v), I_2(v), I_3(v)]^\top$ be its
feature vector, where~$I_\ell(v)=1 \Leftrightarrow v\in V_\ell$,~$0$ otherwise.
Since by design, any cascade generated from~$\scn_1$~($\scn_2$) contains at
least one node from~$V_1$~($V_2$) and no node from~$V_2$~($V_1$), one can easily
distinguish between the two scenarios.
Based on this, we show that a network comprising a trivial single-layer GNN
followed by a single perceptron can learn to distinguish between the two
scenarios.

The GNN layer acts as an aggregator that takes the input graph~(of arbitrary
size) with features. For a cascade~$C$ of size~$n_c$ with adjacency matrix~$A$
and input features~$X\in \reals^{n_c\times 3}$, the output of the GNN layer
is~$H_1=\sigma(AXW_1)$, where~$\sigma$ is the activation function and~$W$ is the
weight matrix. In our case,~$\sigma$ is the identity function and~$W$ is the
trivial~$3\times3$ identity matrix (fixed weights). The resulting output
$H\in\reals^{n_c\times 3}$ is aggregated by summing up across all
nodes:~$H_2=\mathbf{1}^\top H_1\in\reals^{1\times3}$. Here,~$\mathbf{1}$ is the
all ones vector of dimension~$n_c\times1$.

Finally, the perceptron with trainable weights~$W_2\in \reals^{3\times 1}$ is
applied on the output of the aggregator: $H_2W_2$. Note that based on cluster
membership, for~$\scn_1$,~$H_2(1)\ge1$ while~$H_2(2)=0$, and vice versa
for~$\scn_2$. Clearly, the~$H_2$ vectors form a linearly separable space: given
the output variable~$y=\{1,-1\}$~($1$ for~$\scn_1$), a weight
vector~$W^*=[1,-1,0]^\top$, $yH_2W^*>0$. By the Perceptron Convergence
theorem~\shortcite{rosenblatt1958perceptron}, it follows that the problem is
easily learnable.
\end{proof}

\section{GNN MODEL DESCRIPTION}
\label{app:gnn}
\subsection{Message Passing Neural Networks}\label{app:mpnn} An MPNN takes as
its input a graph $G$ in terms of its connectivity information, its node
attributes, and, optionally, its edge attributes. It outputs a latent
representation matrix $\nodeembs\in \mathbb{R}^{|V_G|\times d}$ with row
$\nodeembs_u$ corresponding to a $d-$dimensional latent representation of vertex
$u$. The matrix $\nodeembs$ can also be aggregated to produce a single latent
representation vector for the entire graph $\graphembs\in \mathbb{R}^{d}$. Both
$\nodeembs$ and $\graphembs$ can be used to carry out downstream tasks. 

The operation of a MPNN can be split into two sequential phases, the
message-passing phase and an optional readout phase. At the beginning of the
message passing phase, hidden representations are assigned to the graph nodes,
either using the attribute function $\phi_{V_G}$, or, if nodes are unlabeled,
through some other operation (e.g. random labeling~\shortcite{hamilton2020graph}).
These representations are concatenated to create the initial latent
representation matrix $\hiddenembs^0\in\mathbb{R}^{|V_{G}|\times g}$ where
$\hiddenembs^0_u$ is the $g$-dimensional latent representation of node $u$.
Then, $\hiddenembs^0$ is iteratively propagated through a sequence of
$\gnnnumlayers$ independent layers, with layer $t+1$ passing the information
embedded in $\hiddenembs^t$ further down the graph to produce
$\hiddenembs^{t+1}$, until arriving at the final representation matrix
$\hiddenembs^{\gnnnumlayers}\equiv \nodeembs$. 

More precisely, layer $t+1$ in the message passing phase produces the latent
representation of node $u$ as follows,
\[\hiddenembs^{t+1}_u = \gnnupdate^t(\hiddenembs^{t}_u,
\gnnaggregate_{\gnnsample^t(G, u)}\gnnmessage^t(\hiddenembs^t_{u},
\hiddenembs^t_{v}, \ell_{E}{(u,v)})\] where $\hiddenembs^t_u$ is the latent
representation of node $u$ at layer $t$, $\gnnmessage$ creates a message from
$v$ to $u$, $\gnnaggregate$ aggregates messages from the "neighbors" of $u$
generated from sampling function $\gnnsample^t(G, u)$, and $\gnnupdate$ updates
the latent representation of $u$. Note that, depending on the architecture used,
$\gnnupdate, \gnnaggregate, \gnnsample$ and $\gnnmessage$ may contain learnable
parameters. 

The second phase in a MPNN is the readout phase. It performs the following
operation
\[\graphembs=\gnnreadout(\hiddenembs^\gnnnumlayers)\] where $\gnnreadout$ is an
aggregation function that produces a single graph-level latent representation
$\graphembs$. This phase is only needed when the underlying task is at the
granularity of graphs.

Typically, GNNs are trained through a supervised learning paradigm where the
final latent representations are passed to downstream task layers (such as a
classification layer) and the loss from these layers is backpropagated through
the GNN to update its parameters. In other words, the GNN and the downstream
task layers are trained jointly.

\subsection{Model Overview}\label{app:model} We propose to use a MPNN-based
model to tackle the scenario identification problem. Formally, we create a model
$\module : \cset \to \scnset$ that takes a cascade $\cas$ and predicts a
scenario $\scn \in \scnset$ as the most likely scenario to have produced $\cas$
where $\scnset$ is a set of scenarios $\module$ learns at training time. Model
$\module=(\gnnmodel, \classhead)$ is composed of two components:
\begin{enumerate}
   \item MPNN module $\gnnmodel$ which takes a cascade $\cas$ and produces a
   cascade-level latent representation vector $\graphembs\in \mathbb{R}^{d}$, 
   \item classification head module $\classhead$ which is a Feed-Forward Neural
   Network (FFNN) that takes $\graphembs$ and produces a probability
   distribution for $C$ over the scenario set $\scnset$. 
\end{enumerate}
A model is trained in a supervised manner using a dataset of cascades and their
corresponding scenario labels $\dataset = \{(\cas_1, \ell_1), \dots\}$ where
$\cas \in \cset$ and $\ell_i \in [\scnset]$. Training is done by minimizing the
following empirical loss function:
\[\mathcal{L}(\dataset) = \frac{1}{|\dataset|}\sum_{(\cas, \ell) \in
\dataset}J(\module(\cas), \mathbb{O}_{[\scnset]}^\ell)\] where $J$ is the
cross-entropy loss function. In practice, loss is calculated over a mini-batch
of cascades, and the model is trained using stochastic gradient descent.

Algorithm~\ref{app:alg:train} summarizes the procedure for training module
$\module$. The algorithm takes as input a dataset $\dataset$, as well as a GNN
model $\gnnmodel$ and a classification head $\classhead$, both of which can be
initialized randomly or with pre-trained weights. The algorithm trains the model
by iterating over mini-batches of cascades, calculating the loss for each
cascade, and backpropagating the aggregated loss through the model.

\begin{algorithm}
\caption{\trainfunction}
\begin{algorithmic}[1]
    \REQUIRE training data $\mathcal{D}=\{(\cas_1, \ell_1), \dots\}$; MPNN model
    $\gnnmodel$; classification head $\classhead$\\
    \ENSURE trained model $\gnnmodel$; trained classification head $\classhead$
    \FOR{$b \in $~{\tt miniBatch}($\mathcal{D}$)} \STATE $loss \gets 0$
    \FOR{$\cas, \ell \in b$} \STATE $\graphembs \gets ${\tt
    ApplyGNN}($\gnnmodel,C$) \STATE $loss \gets loss + J(\classhead(\graphembs),
    \mathbb{O}^{\ell}_k)$ \ENDFOR \STATE {\tt backpropagate}($loss$) \ENDFOR
\end{algorithmic}
\label{app:alg:train}
\end{algorithm}

In addition to training the model from scratch, we also consider a transfer
learning setting where a model is trained on a source dataset $\dataset_s$ and
then used as a classifier on a target dataset $\dataset_t$. Following the
taxonomy of \shortcite{Weiss_Khoshgoftaar_Wang_2016}, our approach falls under the
homogeneous transfer learning class. In this setting, the input data of both the
source and target datasets are from the same domain, as well as the output
classes of the source and target datasets. However, the \textit{distribution} of
the target dataset is different from that of the source dataset. For example,
the source and target datasets might be collections of cascades on two different
social networks, but both networks' node and edge features would come from
shared domains, and both datasets would have the same set of scenarios. We
demonstrate using the trained model on $\dataset_t$ directly without any further
training, as well as using the model after fine-tuning it on a subset of samples
from $\dataset_t$.

We demonstrate the transfer learning approach in
Algorithm~\ref{app:alg:higherlevel}. The algorithm takes source and target
datasets as input, as well as an enumerator indicating the type of fine-tuning
that will take place. First, model parameters are initialized according to input
and output data domains. The model is then trained on the source dataset. If the
target dataset is different from the source dataset (signalling that transfer
learning is to be used), a subset of the model parameters are frozen according
to the fine-tuning parameter (preventing them from being updated) and the model
is trained on the target dataset. 

\begin{algorithm}
\caption{\highlevelfunction}
\begin{algorithmic}[1]
    \REQUIRE source dataset $\mathcal{D}_s=\{(\cas_1^s, \ell_1^s), \dots\}$;
    target dataset $\mathcal{D}_t=\{(\cas_1^t, \ell_1^t), \dots\}$; parameters
    to freeze {\tt to\_freeze}\\
    \ENSURE trained target model $\gnnmodel$; trained classification head
    $\classhead$ \STATE $\gnnmodel_i, \classhead_i \gets$ {\tt
    InitializeModel}($\dataset_s$) \STATE $\gnnmodel_s, \classhead_s \gets$
    \trainfunction($\mathcal{D}_s$, $k$, $\gnnmodel_i$, $\classhead_i$)
    \IF{$\mathcal{D}_t \neq \mathcal{D}_s$} \STATE {\tt
    freezeParameters}($\gnnmodel_s$, $\classhead_s$, {\tt to\_freeze}) \STATE
    $\gnnmodel_t, \classhead_t \gets$ \trainfunction($\mathcal{D}_t$, $k$,
    $\gnnmodel_s$, $\classhead_s$) \ENDIF
\end{algorithmic}
\label{app:alg:higherlevel}
\end{algorithm}

\subsection{Architectural Choices for \prob}\label{app:arch} We use an
edge-feature-supported GNN architecture. We define the architecture using the
MPNN terminology described in Appendix~\ref{app:mpnn}.

In the message passing phase, each node receives messages from its incoming
neighbors. The message function is a single-layered feed forward neural network
(FFNN) defined as follows
\[\gnnmessage^t(\hiddenembs^t_u, \hiddenembs^t_v, \ell_E(u,v)) =
concat(\hiddenembs^t_v, \ell_E(u,v))\cdot\mathbf{W}_{message}^t\] where
$\mathbf{W}^t_{message}$ is a trainable weight matrix. If no edge features are
used, the message function becomes:

\[\gnnmessage^t(\hiddenembs^t_u, \hiddenembs^t_v, \ell_E(u,v)) = \hiddenembs^t_v
\cdot\mathbf{W}_{message}^t.\] Incoming messages to a node are sum aggregated.
The update operation is done in two stages. Firstly, we calculate an
intermediate state $\bar{\hiddenembs}_{u}^{t+1}$, 
\[\bar{\hiddenembs}_{u}^{t+1} = \hiddenembs_{u}^t \cdot \mathbf{W}_{self}^t +
\sum_{\{v : (v,u) \in E_C\}}{\gnnmessage^t(\hiddenembs^t_u, \hiddenembs^t_v,
\ell_E(u,v))},\] where $W_{self}^t$ is a trainable weight matrix. Then, we
augment it with batch normalization~\shortcite{ioffe2015batch}, a non-linear
activation function, and an optional skip
connection~\shortcite{dehmamy2019understanding}. Depending on whether a skip
connection is used, the final hidden state of a node is calculated in one of the
three equations shown in Table~\ref{app:tab:update}. 
\begin{table}
    \centering
    \caption{Update functions for the GNN architecture, where
    $BN_{\mathcal{B}}(\cdot)$ is a batch normalization function over training
    batch $\mathcal{B}$.}\label{app:tab:update}
    \begin{tabular}{ll}
        \toprule
        Method & Update function $update(h_{u}^t, a_{u}^t)$\\
        \midrule
        {\tt stack} & $PreLu(BN_{\mathcal{B}}(\bar{\hiddenembs}_{u}^{t+1}))$ \\
        {\tt skipsum} & $\hiddenembs^t_u +
        PreLu(BN_{\mathcal{B}}(\bar{\hiddenembs}_{u}^{t+1}))$ \\
        {\tt skipconcat} & $concat(\hiddenembs^t_u,
        PreLu(BN_{\mathcal{B}}(\bar{\hiddenembs}_{u}^{t+1})))$ \\
        \bottomrule
    \end{tabular}
\end{table}

Additionally, before executing the message passing and readout phases, we
generate node hidden representations from their attributes by passing them
through $R$ FFNNs. Each such layer is followed by a batch normalization layer and
a non-linear activation functions whose architectures match those used in the
message-passing layers. The FFNN layers do not have a bias term due to the batch
normalization operation. The output representations of the $\mathcal{T}$-th
layer of the GNN implicitly comprises the matrix $\nodeembs$.

Since our task is at the graph level, we require every cascade to have a single
latent representation that is agnostic to the number of nodes in the cascade.
For this reason, we add a readout phase to the MPNN architecture that reduces
the embeddings of all the nodes of a cascade into a single vector. The readout
function we use is a summation of the final hidden states of all the vertices of
the cascade. More formally,
\[\graphembs = \sum_{u \in V_C}\hiddenembs^\mathcal{T}_{u}.\] 

Once $\graphembs$ is generated, it is passed to a series of task-specific FFNN
layers. These layers will produce the output used for classification and for
calculating the loss of the model for training, and so we refer to them combined
as the classification head of the model. Classification head layers have bias
terms and each is followed with an activation function, however, they do not
have batch normalization. The output of the final FFNN layer is a vector
$\classheadout \in [0,1]^{|\scnset|}$ comprising a probability distribution over
the set of scenarios $\scnset$ used during training. 

To clarify, the preprocessing and classification-head layers are not part of the
GNN architecture and do not take the connectivity information of the cascade
into account. Each of these layers is shared across all the nodes of all the
cascades in an input batch, but a layer processes each node independently of all
the others.

\subsection{GNN Implementation Details}\label{app:imp}
We use an Adam optimizer, and a PReLU activation function for all layers except
the classification head, which uses the softmax activation. Our GNN is a
4-layers deep {\tt generaledgeconv} architecture in GraphGym, with 300 hidden
units in each layer. For the remaining hyperparameters, we use different flavors
depending on the size of cascades with which the model is being trained. The
differences as well as the values we experimented with are shown
Table~\ref{app:tab:hyperparams}.

Seeding the model was done deterministically; the seed was passed to the
experiment and that seed controlled the process in which nodes were chosen at
the coverage determinism stage as well as the train-test splitting stage.

\begin{table*}[]
    \caption{Hyperparameters used for training GNNs on small and large cascades.
        Small cascades are defined as those capped at a maximum time of 50 days
        and SBM graphs, and large cascades are those with a maximum time of 70
        days.}
    \label{app:tab:hyperparams}
    \centering
    \begin{tabular}{|l|p{1.5in}|p{0.8in}|p{0.8in}|p{0.8in}|}
        \hline
        \textbf{Hyperparameter}        & \textbf{Range of values tried} &
        \textbf{SBM} & \textbf{\popdata{} $T\in\{30, 50\}$} & \textbf{\popdata{}
        $T=70$} \\ \hline
        Message aggregation            & add, mean, max                 & add
        & add                                  & add                        \\
        \hline
        Node feature processing layers & 1, 2, 3                        & 1
        & 1                                    & 2                          \\
        \hline
        Message passing layers         & 2, 4, 6                        & 4
        & 4                                    & 4                          \\
        \hline
        Classification head layers     & 1, 2, 3                        & 2
        & 2                                    & 3                          \\
        \hline
        Skip connection                & skipsum, skipconcat, stack     &
        skipsum      & skipsum                              & stack \\ \hline
        Graph pooling operation        & add, sum max                   & add
        & add                                  & add                        \\
        \hline
        Learning rate                  & 0.0001, 0.001, 0.01            & 0.01
        & 0.001                                & 0.0001                     \\
        \hline
        \end{tabular}
    \end{table*}

\section{FULL EXPERIMENT DESIGN}
\label{app:eval}
We begin this section by describing the two types of cascade datasets we use for
evaluation. The first, described in Appendix~\ref{app:eval:tnva}, corresponds to
cascades simulated on synthetic population contact networks of two U.S. states,
Virginia and Tennessee. The second, described in Appendix~\ref{app:eval:sbm},
corresponds to cascades simulated on random graphs, specifically, Stochastic
Block Model (SBM) graphs. Appendix~\ref{app:eval:models} describes the
baseline model used for evaluation, and Appendix~\ref{app:hwsw} the software
and hardware details.

\subsection{Synthetic Population Datasets (\popdata)}\label{app:eval:tnva}
\subsubsection{Contact Graphs}
We construct two cascade datasets using networks obtained
from~\shortcite{virginiadata}. Each dataset corresponds to a set of cascade graphs
generated through SEIR disease spread simulations on synthetic contact graphs,
one graph is for the state of Tennessee (TN), and the other for the state of
Virginia (VA). Structural parameters of these two networks appear in
Table~\ref{app:tab:graphs}. A contact graph~$G(V,E)$ is an undirected graph that
captures a snapshot of interaction dynamics within a population. Each node in a
contact graph is an individual and its attributes correspond to properties of
that individual such as their age and designation. Each edge is an interaction
between two nodes and it is attributed with a numerical weight equal to the
duration of the interaction, and two categorical labels describing to the
activities of the incident nodes at interaction time. The two categorical labels
can be transformed into a single categorical label describing the nature of
interaction between the two endpoints and can be of the following types:
(i)~\texttt{essential} (like an interaction at home or at work);
(ii)~\texttt{non-essential} (like shopping); or (iii)~\texttt{mixed} (one
individual performing an essential activity, while the other individual a
non-essential activity).   
% TODO: fix the overfull box
\begin{table}[h!]
\caption{\small Structural information about the synthetic contact networks
from Harrison et al. used here. Both networks have a diameter of
length 14.}
\label{app:tab:graphs}
\centering
\begin{tabular}{ccccc}
    Network & |V|       & |E|        & Degree & Degree \\
            &           &            & Max    & Avg.   \\
    TN      & 6,041,517 & 62,149,441 & 461    & 20.6   \\
    VA      & 7,602,717 & 83,162,927 & 543    & 21.9  
    \end{tabular}
\end{table}

%\begin{table*}
%    \centering
%    \begin{tabular}{l|ll}
%        Node attributes & Gender & Values\\
%        \hline
%         & Race & Male, Female\\
%         & Designation &Retail, Medical, Government, Care facilitation,
%         Education, Military, Other \\
%         & Age & Numeric value\\
%         & Age group & Preschool (0-4), Students (5-17), Adults (18-49), Older
%         Adults (50-64), Seniors (65+)\\
%         \hline
%        \hline
%         Edge attributes & Activity & home, work, shop, school, college,
%         religion, other \\
%        \hline
%    \end{tabular}
%    \caption{Node attributes }
%    \label{app:tab:full_atts}
%\end{table*}

\subsubsection{Scenarios}\label{app:tnva:scen} We use four different scenarios
to generate cascades. Scenarios vary by transmissibility, which is a diffusion
model parameter, and two intervention-specific parameters. Transmissibility
$\tau$ correlates with the infection propensity of the disease over all contact
graph edges. The first intervention we consider is vaccination and is defined as
$\inter_{vax}(\sigma_{vax}, \alpha)$. It is a pharmaceutical intervention that
reduces the probability of infection of $|V|\cdot\sigma_{vax}$ randomly selected
(vaccinated) nodes by $\alpha = 80\%$. The second intervention is a generic
social distancing intervention and is defined as $\inter_{gsd}(\sigma_{gsd})$.
It is a non-pharmaceutical intervention that removes all~\texttt{non-essential}
edges incident on $|V|\cdot\sigma_{sgd}$ randomly selected (compliant) nodes.
The parameters corresponding to each scenario are shown in
Table~\ref{tab:design}. The diffusion process parameters and scenario parameters
were chosen to ensure that the cascades generated under different scenarios are
not easily distinguishable by cascade measures such as the number of infections
per day, or overall cascade sizes. 

% MOVED to main text (Table~\ref{tab:design}); gated so it vanishes under
% \showmovedfalse. Float cannot use the movedout environment, so we gate with \ifshowmoved.
\ifshowmoved
\begin{table}[h!]
\caption{$^*$Scenarios on the VA network used  $\tau = 0.155$. $^\dag$Scenario on the VA
network used $\sigma_{gsd} = 60\%$.}
%%\aaa{for the TN contact network}.}
\label{app:tab:design}
\centering
\begin{tabular}{ |c|c|c|c| }
\hline
\textbf{Scenario} & $\tau$ & $\sigma_{vax}$ & $\sigma_{gsd}$\\
\hline
No vax + Low GSD & 0.09 & None & 25\%\\
No vax + High GSD & 0.09 & None & 70\%\\
Vax + Low GSD & $0.16^*$ & 50\% & 25\%\\
Vax + High GSD & $0.16^*$ & 50\% & $65\%^\dag$\\
\hline
\end{tabular}
\smallskip
\end{table}
\fi

\subsubsection{Cascade Generation}
Cascade simulation proceeds as follows: given the contact graph and a scenario,
the simulator applies vaccination and SGD interventions to randomly selected
nodes according scenario parameters, it randomly selects 20 nodes to be
initially infected nodes (i.e. seeds), and allows the disease to propagate
through the contact network for 300 days. We generate 100 cascades per scenario.
Thus, each of the TN and VA data sets includes 400 cascade graphs. From each set
of 400 cascades, we produce three datasets, each with a different time horizon.
We use time horizons $T\in\{30,50,70\}$ for our experiments, as it was shown
that for this setting, scenarios are less distinguishable earlier in the
simulation~\shortcite{harrison2023identifying}. To clarify, setting the time horizon
of a cascade to $t$ means removing all infections which occurred at time $t'>t$.
Cascades are generated using the EpiHiper epidemic
simulator~\shortcite{chen2025epihiper}.

\subsubsection{Data Representation}\label{app:eval:tnva:datarep} Nodes have a
numerical age value, as well as categorical values for their age group, race,
gender, and designation. Additionally, we augment for each node $u$ in the
cascade their boundary degree (as defined in Section~\ref{sec:node_rep}).

Unless otherwise stated, edges are attributed with the concatenation of two
one-hot-encoded vectors, each of size 7, corresponding to the activity of the
source and destination nodes, respectively.

%We experiment with three different alternatives of edge feature representation
%that vary in terms of the granularity of the information represented:
%\begin{itemize} \item \noedge: we do not use edge features in the
%classification procedure. \item \threeway{}: we use a single one-hot-encoding
%vector that labels an edge as~\texttt{essential}, ~\texttt{non-essential}, or
%~\texttt{mixed}, as described previously in this section. \item \fourteenway{}:
%we concatenate two one-hot-encoded vectors, each of size 7, corresponding to
%the activity of the source and destination nodes, respectively.  
%\end{itemize} 
\subsection{Stochastic Block Model (SBM) Dataset}\label{app:eval:sbm}
\subsubsection{Contact Graphs}
We carry out multiple empirical studies on cascade datasets generated on random
SBM graphs. In these graphs, nodes are unattributed while edges are attributed
with two categorical labels, one for each endpoint. The labels correspond to the
community of the corresponding node.
\subsubsection{Scenarios}
Scenarios are described by the parameters used to generate the underlying SBM
contact graph and the parameters of the contagion model. We do not use any
explicit interventions in our scenario definitions, though different scenario
parameters can be interpreted as applying different interventions. As a
diffusion model, we use an independent cascade (IC) model. We define a scenario
$\scn=(\sbm(n, \pi_k, W), \icmodel_{\sbm}(\mathbf{T}, \nu))$  by its SBM
parameters as described in Section~\ref{sec:preliminaries} as well as its
diffusion model parameters $\mathbf{T} \in [0,1]^{k \times k}$ and $\nu$ which,
combined, define transition probabilities. To elaborate, given these parameters,
an infected node $u$ in community $i$ infects a susceptible node $v$ in
community $j$ with probability $T_{ij} \cdot \nu$.

\subsubsection{Cascade Generation}
Given a scenario tuple $(\scn_1, \dots)$, we generate for each scenario an SBM
graph and 1000 cascades using the parameters of the scenario. We run the
diffusion process until completion without limiting the time horizon of the
cascades at any fixed point. SBM graphs are generated using the NetworkX
library~\shortcite{hagberg2008exploring}, and independent cascade simulations are
carried out using the EoN package~\shortcite{Miller2019}.

\subsubsection{Data Representation}
Nodes are attributed with their degree in the contact network, their degree in
the cascade, and their boundary degree as defined in Section~\ref{sec:node_rep}.
GNNs require every node to have a feature vector. However, for some of our
experiment, we experiment with the dataset without using the aggregated
features. In such cases, we must add arbitrary features. We experiment with
multiple known workarounds for unattributed graphs, which includes using
randomly sampled features, random one-hot-encoded vectors, and identity feature
vectors~\shortcite{cs224wl3,hamilton2020graph}. We find that performance is not
significantly affected by the  method. We decide to label all nodes with the
identity feature vector $= \{1\}^{30}$ as it was the most efficient solution. As
for edge attributes, we use one of three alternatives to label edges.
\begin{itemize}
    \item \noedge{}: defined exact as in Appendix~\ref{app:eval:tnva:datarep}.
    \item \communitymatch{}: a one-hot-encoded vector that indicates if the
    endpoints of an edge belong to the same community.
    \item \communityid{}: two one-hot-encoded vectors corresponding to the IDs
    of the communities of the endpoint nodes.
\end{itemize}

\subsubsection{Dataset Instantiations}
We experiment with two different SBM datasets, \sbmstudy{} and \sbmthreecomm{}.

\subsection{SBM Dataset Instantiations}\label{app:sbm}
\subsubsection{\sbmstudy{}}\label{app:sbm:study} is a set of 1128 datasets, each
corresponding to two \sbm{} scenarios. This dataset is meant to represent a very
wide range of cross-scenario structural variety. All scenario pairs have
$n=2000$, equisized community sizes, and symmetric $\mathbf{W}$ and $\mathbf{T}$
matrices, but the values of $\mathbf{W}$, $\mathbf{T}$ and $\nu$ vary across
scenarios. Scenario pair parameter selection is done as follows: we select a
range of possible values for the inter- and intra-community SBM edge
probabilities and IC transition probabilities, as well as a range for the
transmissibility parameter $\nu$.  We then generate 1000 cascades for all
possible scenarios by taking the Cartesian product of these ranges. Initially,
we add all possible scenario pairs from this set of scenarios to \sbmstudy{}.
Then, for each pair of scenarios, we exclude it from \sbmstudy{} if 
\begin{enumerate*}[label=(\roman*)]
    \item for either scenario, the mean number of infected edges in the all the
scenario's cascades is less than $5\%$ of the total nodes in the underlying
graph, or
    \item the mean number of infections of the cascades of one scenario is more
    than $6\times$ that of the other scenario. 
\end{enumerate*} 
The ranges used for generating the scenario pairs are in
Table~\ref{table:sbmstudy}.
\begin{table*}[t]
    \caption{Parameters used to generate the \sbmstudy{} dataset.}
    \label{table:sbmstudy}
    \centering
    \begin{tabular}{|p{2.4in}|p{2.4in}|p{1.5in}|}
    \hline
    \textbf{Parameter}                & \textbf{Description}
    & \textbf{Values}                                \\ \hline
    SBM inter-community probability   & Probability scaler in the SBM graph that
    an edge forms across communities, i.e. $W_{ij} : i\neq j$
    & $\{0.005,0.001\}$                              \\ \hline
    SBM intra-community probability   & Probability scaler in the SBM graph that
    an edge forms within a community, i.e. $W_{ii} $
    & $\{0.031, 0.053, 0.073\}$                      \\ \hline
    SIR inter-community infectiousness & Probability scaler in the SIR diffusion
    that an infection occurs from community to the other, i.e. $\mathbf{T}_{ij}
    : i\neq j $ & $\{0.5, 0.3, 0.1\}$                            \\ \hline
    SIR inter-community infectiousness & Probability scaler in the SIR diffusion
    that an infection occurs within a community, i.e. $\mathbf{T}_{ii}$
    & 1                         \\ \hline
    Transmission probability          & Global infection transmission
    probability $\nu$ & $\{0.029, 0.047, 0.071\}$ \\ \hline
    \end{tabular}
    \end{table*}

\subsubsection{\sbmthreecomm{}} is a set of datasets each corresponding to two
scenarios of three communities each. Scenario pairs in this dataset are designed
in such a way that for the pair $(\scn_1, \scn_2)$, cascades from one scenario
are indistinguishable from the other without community membership information.
That is done by guaranteeing that $\pr(\cas \in \scn_1) = \pr(\cas' \in
\scn_2)$, where $\cas'$ is equivalent to $\cas$ but with the first and second
communities' nodes swapping communities. The dataset is defined as follows: all
scenarios have $n=3000$, with three communities of equal sizes. They share the
SBM matrix $\mathbf{W}$ and use the structural transition matrices
$\mathbf{T_1}, \mathbf{T_2}$ given in the main text (Section~\ref{sec:eval:sbm}).
% MOVED to main text (Section~\ref{sec:eval:sbm}); shown here under \showmovedtrue.
This translates to the first two communities having identical
inter-community transition probabilities across both scenarios, while their
inter-community transition probability with the third community is the exact
opposite across scenarios. All scenario pairs in \sbmthreecomm{} have the same
$\alpha$, but differ in $s$, which falls in the range described above.

\subsection{Baselines}
\label{app:eval:models}

\subsubsection{Feature Engineering (\epicurveandstructure{})}
We use the model proposed in~\shortcite{harrison2023identifying} as the
feature-engineering baseline. This model uses aggregated structural features
of the cascade to do a closed-world classification of the scenario.
Specifically, it uses the counts of different types of paths that appear within
each cascade as features. For example, one feature used was the number of mixed
edges, and another was the number of paths of length two where one edge was
nonessential, and one edge was essential. It also made use of boundary
information from the cascade. Specifically, it counts the total number of edges
in the underlying graph $G$ incident to, but not contained within, the cascade
$\cas$. This model also made use of the histogram of outdegrees within the
cascade. By using aggregated structural features, the authors were able to
achieve good performance using a variety of model types (random forests,
logistic regression, and support vector machines). Moving forward, we
will refer to the baseline model as \epicurveandstructure{}.

\subsection{Hardware and Software Description}\label{app:hwsw} The GNNs used in
our work are all implemented in Python using PyTorch Geometric (PyG)
v2.3.1~\shortcite{Fey/Lenssen/2019}. Specifically, we use the
GraphGym~\shortcite{you2020design} library, which is a wrapper on top of PyG that
allows for easy experimentation with GNN architectures. We use the GraphGym
version prepackaged with PyG. We use Python 3.9.13 to run our code. Generating
the SBM graphs is done using NetworkX 3.1~\shortcite{hagberg2008exploring}. Discrete
SIR diffusion processes that were done on SBM graphs were generated using EoN
3.1~\shortcite{Miller2019}. SEIR simulations over the population networks is done
using the EpiHiper epidemic simulator~\shortcite{chen2025epihiper}.

The servers we use to run our experiments use Rocky Linux 8.9 (Green Obsidian)
as their OS. We use multiple different servers and they have between 100GB to
1000GB of memory. Depending on the experiment, we may use 32GB to 80GB capacity
GPUs. Exact library versions can be found in the supplementary material package.

\section{ADDITIONAL RESULTS}
This section presents additional experiments that complement the main results.

\begin{figure}
    \includegraphics[width=\linewidth]{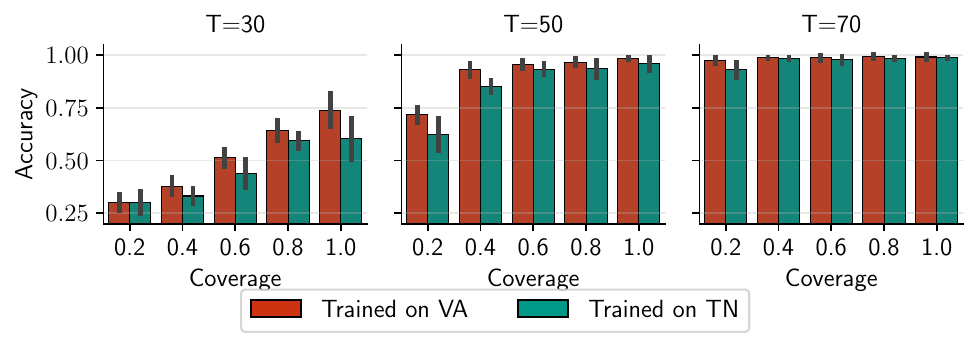}
    \caption{Transfer learning: accuracy on VA for a model trained on VA
    (red) and a model trained on TN with matching time and coverage (green).}
    \label{app:fig:transfer_matching}
\end{figure}

\begin{figure}
    \includegraphics[width=\linewidth]{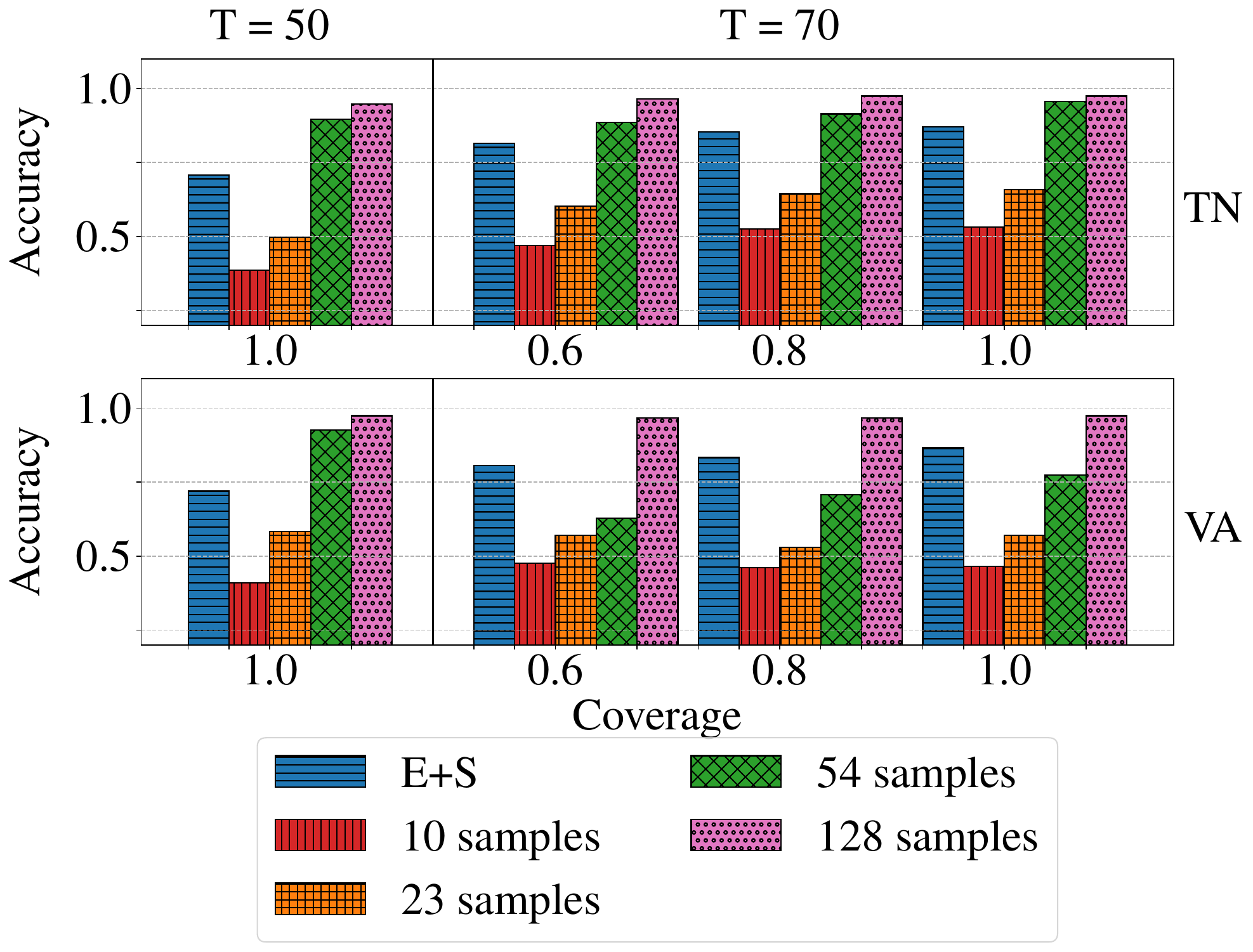}
    \caption{Accuracy of the GNN model compared on TN at different times and
    coverage values using a range of training sample sizes. The performance of
    the GNN is compared with the GNN \epicurveandstructure{} model on the same
    dataset. Note that the \epicurveandstructure{} was trained on 300 samples.}
    \label{app:fig:sample_sens}
\end{figure}

Additionally, we carry out an experiment comparing the performance of
\epicurveandstructure{} with the GNN model using different amounts of training
samples. Results are shown in Figure ~\ref{app:fig:sample_sens}. We find that
the GNN model is able to surpass the \epicurveandstructure{} on TN using less
than  $20\%$ of the number of samples, while on VA at time 70, it requires about
$42\%$ of the number of samples.

\subsection{Transfer Learning}\label{app:eval:transfer}
In addition to the matching-distribution transfer results shown in the main
body (Figure~\ref{app:fig:transfer_matching}), we examine how well the learned
models transfer across mismatched time horizons and coverage levels.

\begin{figure*}
\includegraphics[width=\linewidth]{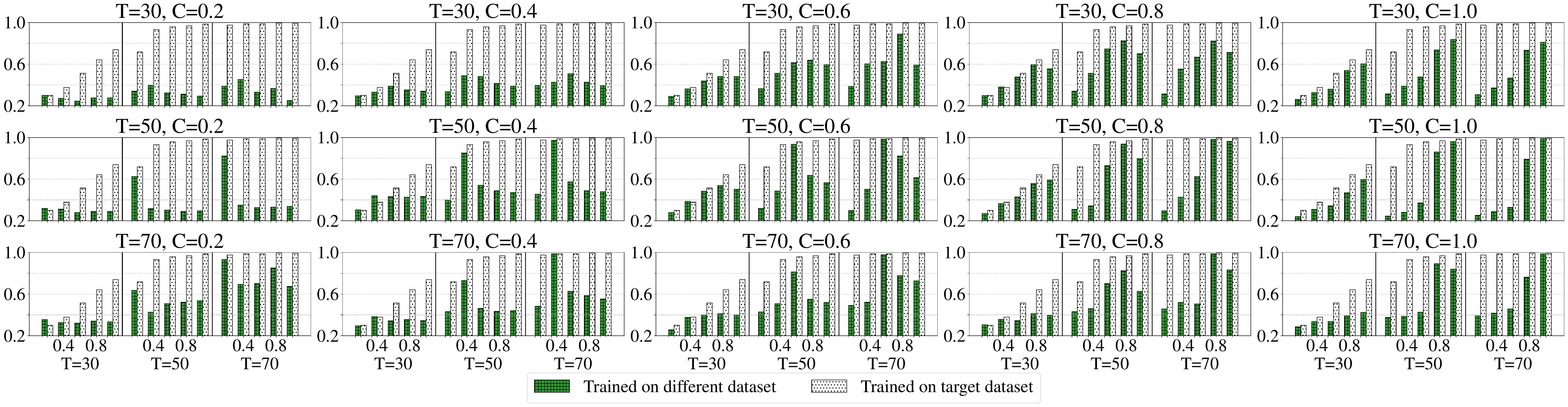}
\caption{Accuracy on all VA datasets when the model is trained on the evaluation
dataset versus when it is trained on a TN dataset with the properties shown in
each subfigure's title. In other words, each subplot shows the accuracies
achieved by a single model trained on TN on all VA datasets. }
\label{app:fig:transfer_nomatching}
\end{figure*}
Next, we wanted to examine how well the learned models transfer across two
additional axis of the data distribution; namely, time and coverage. In other
words, how well can a model trained on certain time and coverage values on
network some generalize to other datasets with looser, or stricter, time or
coverage constraints, and on the other network. This is an additional dimension
of robustness that we envision would be important for scenario identification.
Results for this experiment are shown in
Figure~\ref{app:fig:transfer_nomatching}. In each plot, green bars are
accuracies of a single model trained on TN, on the time and coverage values
mentioned in the title of the corresponding plot. White values are accuracies of
models trained on the target dataset. We find that models perform better when
the distribution of the target dataset is closest to the distribution used to
train them. This can be seen in the accuracies on target datasets that match a
model's source dataset's time and/or coverage, and the pattern in which accuracy
degrades as the difference in these values between source and target datasets
increases.

\begin{figure}
    \includegraphics[width=\linewidth]{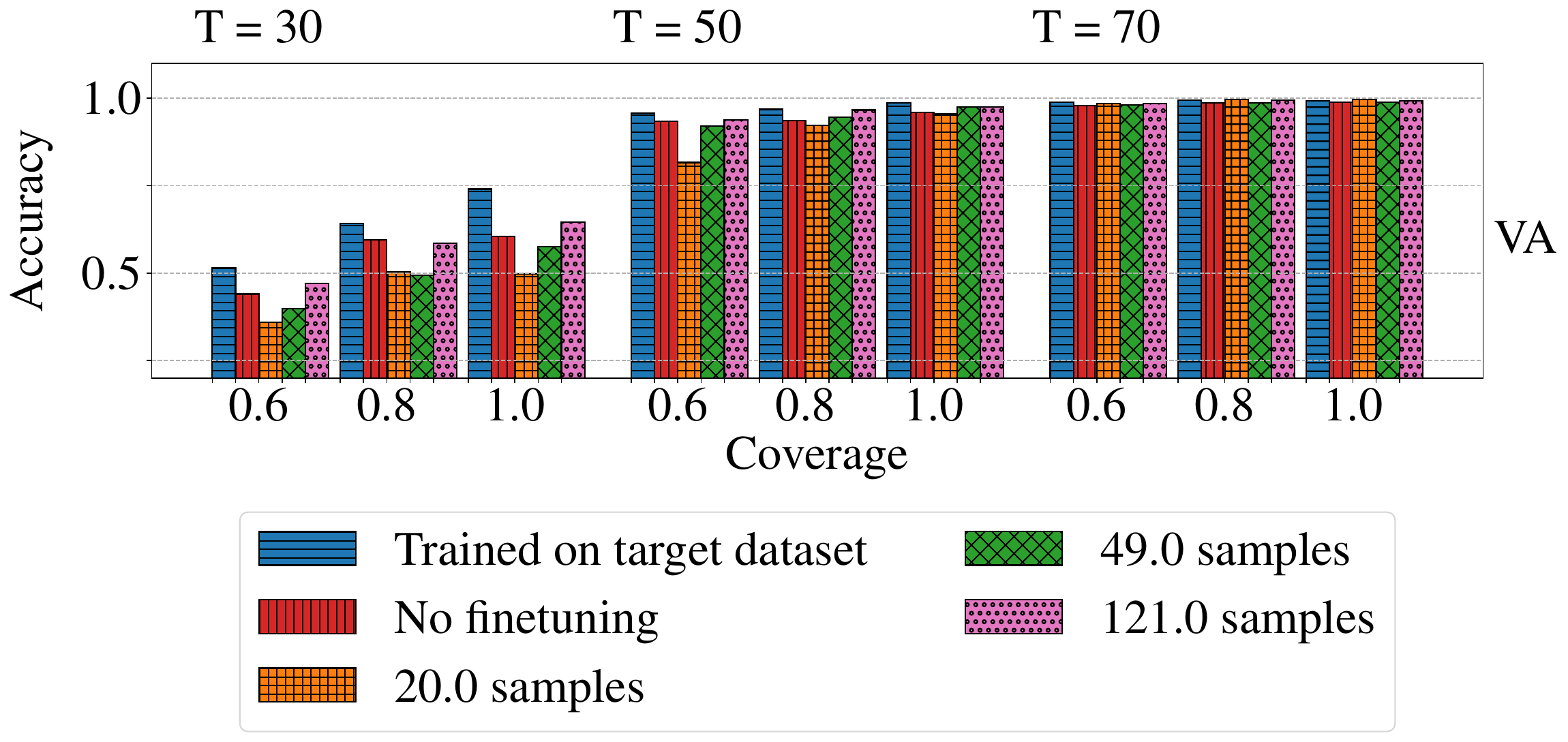}
    \caption{Results from transferring a model trained on TN to VA by
    fine-tuning its classification head $\classhead$ on a sub-sample of examples
    from the target set. Note that models are trained on 300 samples.}
    \label{app:fig:fine_tune_va}
\end{figure}

We run an experiment with fine-tuning the model. We train the model on TN and
fine-tune it on a subsample of VA. Specifically, we fine-tune the
classification head only $\classhead$. Results are shown in
Figure~\ref{app:fig:fine_tune_va}. We observe that, as shown in previous results
(Figure~\ref{app:fig:transfer_matching} and Figure~\ref{app:fig:transfer_nomatching}),
using a model trained on TN directly on VA without fine-tuning performs close to
a model trained on VA itself (first and second columns). Fine-tuning the model
can take up to 121 samples before performance improves over the non-finetuned
version. However, this can be beneficial in terms of computation cost. For
example, a single training epoch of fine-tuning the classification head on 121
samples from the VA dataset at time 50 and with full coverage is $4.9\times$
faster than training the full module $\module$ on all 300 samples, and the
resulting model is only $1\%$ less accurate. 

\subsection{Separability Detection}\label{app:eval:sep} In this experiment
(results shown in the right panel of Figure~\ref{fig:sbm} in the main body), we use the
\sbmthreecomm{} dataset and vary the value of the separability parameter $s$ in
the range $[0.5, 1]$
using the data representation \communityid{}. The value of $s$ controls the
difference in the number of edges between community 3 and the other two
communities across scenarios. When $s=0.5$, the two scenarios have identical
distributions; when $s=1$, one scenario has no edges between community 3 and
community 1, while the other has no edges between community 3 and community 2.

\subsection{Additional Ablation and Robustness Analysis}\label{app:eval:edge}

\begin{figure}[t]
    \includegraphics[width=\linewidth]{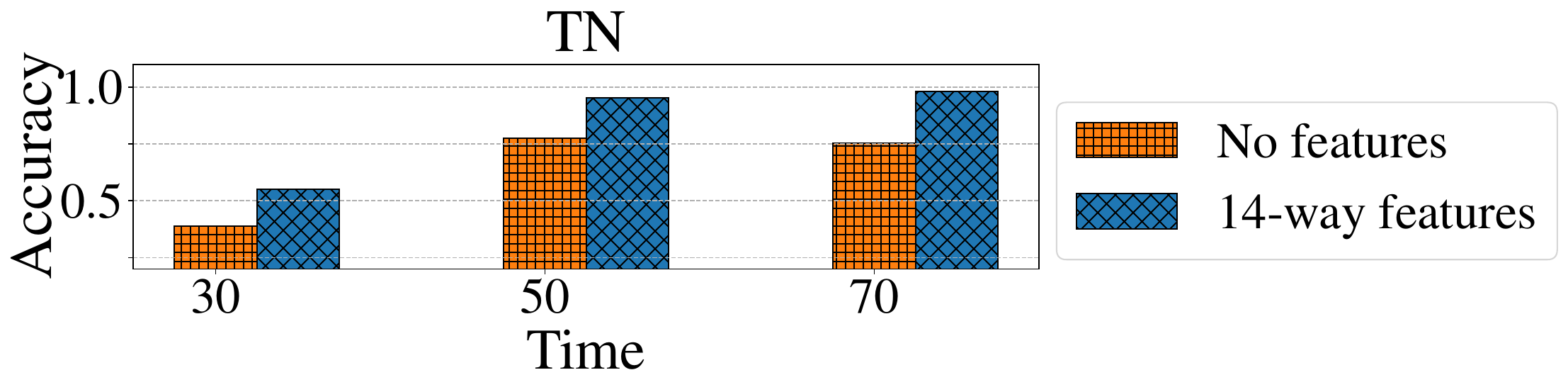}
    \caption{Comparing the accuracy of the GNN model on TN with and without
    using edge features. The GNN with 14-way edge features
    (\texttt{generaledgeconv}) consistently outperforms the same architecture
    without edge features (\texttt{generalconv}).}
    \label{app:fig:noedge_vs_14way}
\end{figure}

Figure~\ref{app:fig:noedge_vs_14way} shows the effect of edge features on TN
across time horizons and coverage levels at full coverage.

To further emphasize the importance of edge features, we run an experiment on
\sbmthreecomm{} where we train a model using the \noedge{} data representation
(where community IDs are not known to the model) using all values of $s$. We
find that, regardless of the value of $s$, the model always behaves like a
random classifier. This confirms that the scenario pairs are designed so that
differences between them are only measurable when community IDs are available.

\begin{figure}
    \includegraphics[width=\linewidth]{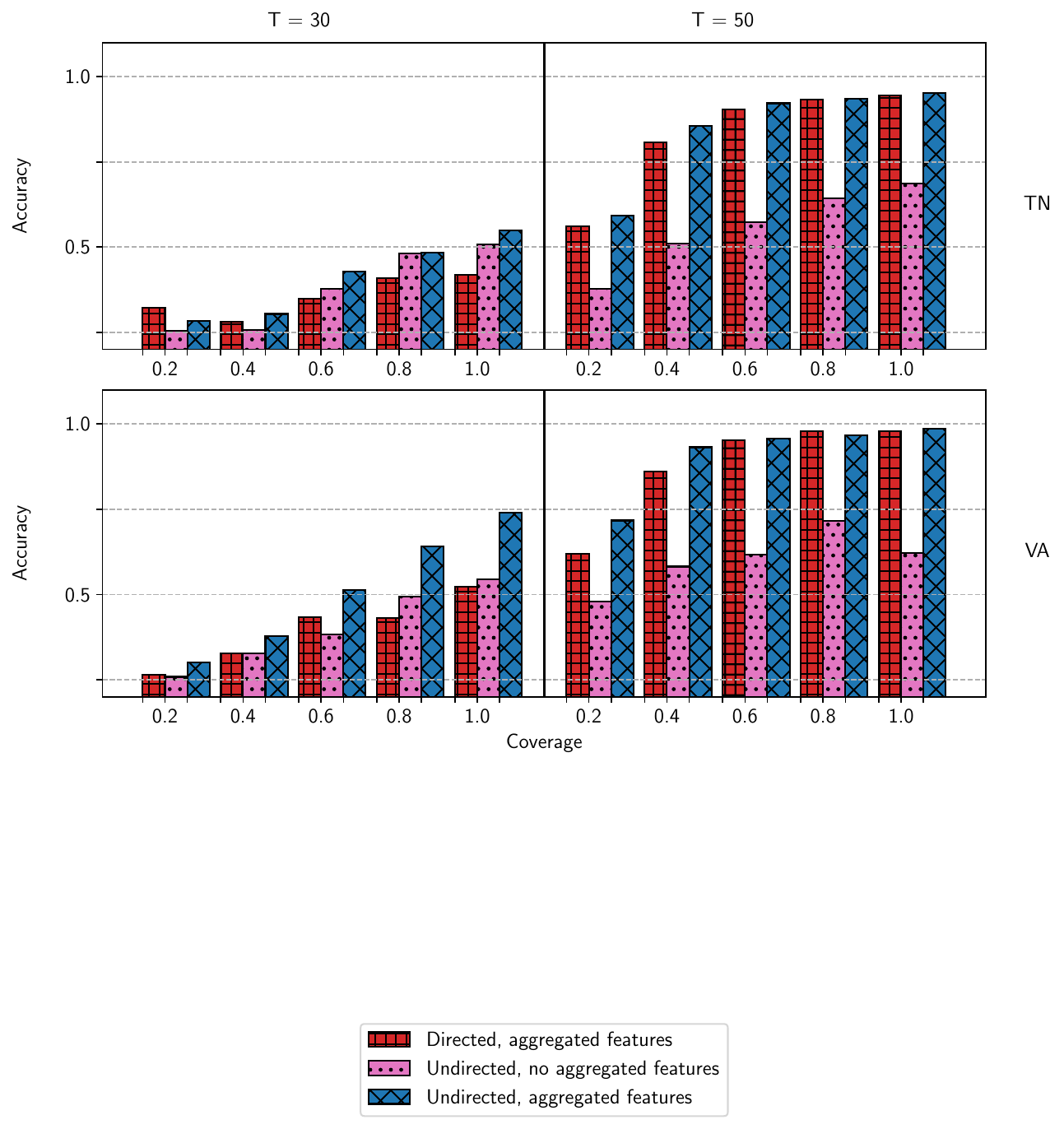}
    \caption{Comparison of GNN accuracy when incorporating boundary edge counts
    and when the graph is undirected.} \label{app:fig:boundary_direction}
\end{figure}
% Please add the following required packages to your document preamble: %
% \usepackage{multirow}
\begin{table}[ht]
  \centering
    \caption{The accuracy of the GNN model when not using aggregated features,
    when using each of the three aggregated features, and when using all three
    features. Experiment was run at time 50 with full coverage, and the cascades
    were directed. } \label{app:table:boundary}
    \begin{tabular}{ll|rr}
        Dataset          &                        & TN   & VA   \\\hline
        Accuracy (\%) & No aggregated features & 63.2 & 54.6 \\
                      & Cascade degree         & 70.4 & 66.2 \\
                      & Graph degree           & 66.0 & 48.4 \\
                      & Boundary degree        & 92.4 & 97.0 \\
                      & All                    & 94.4 & 97.8
        \end{tabular}
    \end{table}
\subsubsection{Degree Counts} 
We measure the benefit of the locally aggregated node features, namely nodes'
graph, cascade, and boundary degrees, on the model performance.
Figure~\ref{app:fig:boundary_direction} and the right subfigure of
Figure~\ref{fig:sbm} demonstrate the performance gain from these features on
\popdata{} and \sbmstudy{}, respectively. On average, boundary degree improves
\popdata{} accuracy by 19\% and \sbmstudy{} accuracy by 42\%. We carry out
an additional ablation study over the three features to understand their effect
more finely. Table~\ref{app:table:boundary} shows results on a subset of
\popdata{} at $T=50$ with full coverage. We find that the majority of the
performance gain can be attributed to the boundary degree. On this subset it
improves the model's accuracy by an average of 35.8\%, and adding the other two
features contributes only a further 1.4\%. We attribute this to the fact that boundary
degrees combine information from the cascade as well as from the contact
network, revealing important information about the cascade relative to its
underlying network. We note that these results follow what we had already shown
theoretically at the end of Section~\ref{sec:learnability} in terms of boundary
degree utility.

\subsubsection{Non-directionality}\label{app:eval:direction} Cascades are
directed graphs; infections flow in one direction. However,
Figure~\ref{app:fig:boundary_direction} shows the accuracies achieved when
treating the graph as a directed graph (red) versus treating it as an undirected
graph (blue). We find that, consistently, using a directed representation yields
better results. On average, we find a 6\% increase in accuracy achieved by
making the graph undirected. We attribute this to the better flow of information
through the cascade unlike in the directed representation in which nodes only
get information from their ancestors, and disseminate information to their
descendants.  
\begin{figure}
    \includegraphics[width=\linewidth]{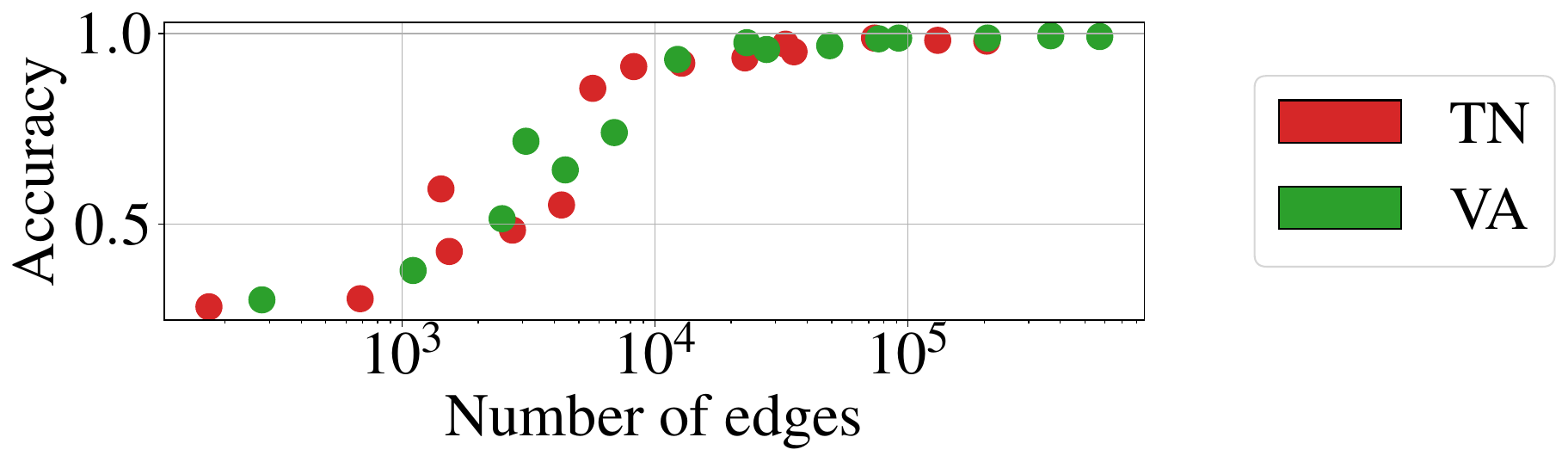}
    \caption{Plotting the accuracy of the GNN model versus the mean number of
    edges of the cascade dataset being classified. Specifically, all the
    datasets that are classified in Figure~\ref{fig:gnn_vs_es} are plotted.}
    \label{app:fig:size}
\end{figure}
\subsubsection{Large Cascade Sizes} Naturally, larger cascades contain more
information than smaller cascades, and so more can be learned (and inferred)
from larger cascades. To showcase the effect of cascade size on achievable
accuracy, we model all the accuracies over all the datasets shown in
Figure~\ref{app:fig:size}, but with the $x$-axis showing the mean size of the
cascades in that dataset. We see a clear correlation between the two.

\fi

% Reducing font size for References & Author Biographies
\footnotesize

\bibliographystyle{wsc}
\bibliography{material/references}

\section*{AUTHOR BIOGRAPHIES}

\noindent {\bf \MakeUppercase{Amro Alabsi Aljundi}} is a Ph.D. candidate at the
Biocomplexity Institute, University of Virginia. His research focuses on machine
learning and networks. His e-mail address is \email{nmm2uy@virginia.edu}.\\

\noindent {\bf \MakeUppercase{Galen Harrison}} is a researcher at the Biocomplexity Institute, University of Virginia. His research interests include machine learning for epidemic simulation analysis and scenario identification. His e-mail address is \email{gh7vp@virginia.edu}.\\

\noindent {\bf \MakeUppercase{Jiangzhuo Chen}} is a Research Professor at the
Biocomplexity Institute, University of Virginia. He is a lead developer of
EpiHiper, an epidemic simulator used in national-scale public health studies. His e-mail address is \email{chenj@virginia.edu}.\\

\noindent {\bf \MakeUppercase{Abhijin Adiga}} is a Research Associate Professor at the Biocomplexity Institute, University of Virginia. His research interests include network science, combinatorial optimization, and computational epidemiology. His e-mail address is \email{abhijin@virginia.edu}.\\

\noindent {\bf \MakeUppercase{Anil Kumar Vullikanti}} is a Professor at the Biocomplexity Institute, University of Virginia. His research spans algorithm design, network science, and machine learning with applications in public health. His e-mail address is \email{vsakuma@virginia.edu}.\\

\noindent {\bf \MakeUppercase{Madhav V. Marathe}} is a Distinguished Professor and Director of the Biocomplexity Institute at the University of Virginia. His research interests include simulation and modeling of large socio-technical systems, computational epidemiology, and network science. His e-mail address is \email{marathe@virginia.edu}.\\

\end{document}